\documentclass[journal,twocolumn]{IEEEtran}
\usepackage{amsmath,amssymb,amsthm,mathtools,bm}
\usepackage{graphicx}
\usepackage{microtype}
\usepackage{xcolor}
\usepackage{hyperref}
\usepackage{booktabs}

\IfFileExists{pcrr7t.tfm}{}{}

\mathtoolsset{showonlyrefs}

\hypersetup{
  colorlinks=true,
  citecolor=blue!45!black,
  linkcolor=blue!45!black,
  urlcolor=blue!45!black,
  pdftitle={Exact Critical Curve for Uniform Stabilizer-State Identification}
}

\newtheorem{theorem}{Theorem}
\newtheorem{lemma}[theorem]{Lemma}
\newtheorem{proposition}[theorem]{Proposition}

\theoremstyle{definition}

\newtheorem{remark}[theorem]{Remark}
\theoremstyle{plain}

\newcommand{\F}{\mathbb F}
\newcommand{\Stab}{\mathsf{Stab}}
\newcommand{\Tr}{\operatorname{Tr}}
\newcommand{\rank}{\operatorname{rank}}
\newcommand{\Fix}{\operatorname{Fix}}
\newcommand{\ran}{\operatorname{ran}}
\newcommand{\Ord}{\mathrm{ord}}
\newcommand{\Exc}{\mathrm{exc}}
\newcommand{\Sp}{\operatorname{Sp}}
\newcommand{\St}{\operatorname{St}}
\newcommand{\id}{\operatorname{id}}
\newcommand{\one}{\bm 1}

\newcommand{\cG}{\mathcal G}

\newcommand{\gbinom}[3]{\genfrac{[}{]}{0pt}{}{#1}{#2}_{#3}}

\begin{document}
\title{Exact Critical Curve for Uniform Stabilizer-State Identification}
\author{Masahito Hayashi$^\dagger$, Yimin Lu$^{\ddagger}$
\thanks{$^{\dagger,\ddagger}$ The authors are listed in alphabetical order.}
\thanks{M. Hayashi is with the School of Data Science, The Chinese University
of Hong Kong, Shenzhen, China; the International Quantum Academy, Shenzhen,
China; and the Graduate School of Mathematics, Nagoya University, Japan
(e-mail: hmasahito@cuhk.edu.cn).}%
\thanks{Y. Lu is with 
Yau Mathematical Sciences Center, Tsinghua University,
Beijing 100084, China.}%
\thanks{M. Hayashi was supported in part by the Guangdong Provincial Quantum
Science Strategic Initiative (Grant No. GDZX2505003) and the General R\&D
Projects of 1+1+1 CUHK-CUHK(SZ)-GDST Joint Collaboration Fund (Grant No.
GRDP2025-022).}}

\maketitle

\begin{abstract}
The sample complexity of identifying an unknown pure stabilizer state has been a longstanding problem in quantum information. In this paper, we resolve the entire exact critical crossover for the uniform ensemble of pure $n$-qubit stabilizer states under
arbitrary collective measurements and without side information.  For
$k=n+c$, with fixed integer $c$, the optimal success probability converges to
an explicit positive series $P_{\infty}(c)$.  In particular,
$P_{\infty}(0)=0.1760837896\ldots$, while its two tails obey a
quadratic-exponential law below threshold and
$1-P_{\infty}(c)\sim2^{-c}$ above threshold.  The pretty-good measurement is
exactly optimal for every finite $n$ and $k$.  The proof diagonalizes the
stabilizer Gram operator through Clifford tensor-power duality and a defect
filtration, including the exceptional binary sectors.  The resulting
integer-scale law quantifies the $n+O(1)$ copy threshold.
\end{abstract}


\section{Introduction}
\label{sec:introduction}
Stabilizer states are the discrete backbone of quantum error correction,
Clifford computation, and many certification protocols
\cite{AaronsonGottesman2004,Gross2006,GrossNezamiWalter2021}.  Their number on
$n$ qubits,
\begin{equation}
 N_n=2^n\prod_{j=1}^{n}(2^j+1),
 \label{eq:number}
\end{equation}
is already $2^{n^2/2+O(n)}$.  It is therefore natural to ask how many copies
of a uniformly random stabilizer state are needed to recover its \emph{full}
label: the projective state, equivalently its stabilizer group including the
Pauli signs.  We assume no classical side information.  Here we determine
the optimal success probability throughout the integer-scale window
$k=n+c$ for qubits, thereby quantifying the $n+O(1)$ copy threshold.

Hayashi \emph{et al.} established the first-order copy-rate threshold and
explicitly left the bounded-offset window $k=n+O(1)$
unresolved~\cite[Remark~10.6 and Sec.~13.3]{HayashiLu2026}.
At the complete-identification endpoint, their results give vanishing
success at every fixed copy rate below one and success tending to one
when $k-n\to+\infty$, but do not determine the success probability at
any fixed offset.  We resolve this open problem for the uniform ensemble
of all pure $n$-qubit stabilizer states.  The central step is an exact
determination of the complete finite spectrum of the stabilizer Gram
operator, including the eigenvalues and multiplicities of both ordinary
and exceptional sectors.  From this spectrum we derive the exact
finite-copy optimum and, with uniform control of the spectral sum, the
exact constant-offset limiting curve $P_\infty(c)$ for every fixed
integer $c$ in $k=n+c$.

Related learning results address the linear copy scale and
computational efficiency.  Montanaro's Bell-sampling algorithm identifies
pure stabilizer states from $O(n)$ copies with exponentially small failure,
using measurements on at most two copies at a time and $O(n^3)$ classical
processing time~\cite{Montanaro2017}.  More recently, Chang \emph{et al.}
proved finite-copy bounds for identification with worst-case failure
probability at most $\delta$~\cite{ChangCaroLaroccaWest2026}.  For
$0<\delta<1/8$, their minimum-copy quantity $L_\delta(n)$ satisfies
\begin{equation}
 n+\lceil\log_2(1/\delta)\rceil-3
 \leq L_\delta(n)
 \leq n+\lceil\log_2(1/\delta)\rceil+4,
 \label{eq:related-copy-bounds}
\end{equation}
and the upper bound admits a polynomial-time implementation.

For the uniform ensemble considered here, the covariant pretty-good measurement (PGM)
has the same conditional success for every
unknown state; consequently this value is also minimax optimal,
as shown in Section~\ref{sec:measurement}.  The distinction from the
finite-copy bounds above is the exact operational curve and its two tail
constants.  Efficient implementation of the optimal measurement remains
a separate question.

\section{Model and main results}
\label{sec:results}
Let $\Stab_n$ denote the set of projective pure $n$-qubit stabilizer states,
and write $\rho_x=|x\rangle\!\langle x|$ for $x\in\Stab_n$.  For $k$
identical copies define
\begin{equation}
 P^*_{n,k}:=\max_{\{M_x\}}
 \frac{1}{N_n}\sum_{x\in\Stab_n}
 \Tr(M_x\rho_x^{\otimes k}),
 \label{eq:operational}
\end{equation}
where the maximization is over arbitrary collective POVMs
$\{M_x\}_{x\in\Stab_n}$.  Our main theorem determines this probability
exactly for every finite number of qubits and copies by giving the complete
Gram spectrum.  Its fixed-offset consequence, Theorem~\ref{thm:curve},
then describes the critical window $k=n+c$.

Choose a unit vector representative for each projective state and let
$H_{n,k}$ be the $N_n\times N_n$ Gram matrix with entries
$(H_{n,k})_{xy}=\langle x|y\rangle^k$.  Changing these representatives
conjugates $H_{n,k}$ by a diagonal unitary and hence leaves its spectrum
unchanged.  Write $D:=2^n$ and define the moment normalization
\begin{equation}
 Z_{n,k}:=D\prod_{j=0}^{k-2}(D+2^j).
 \label{eq:main-Z}
\end{equation}
All empty products in this paper equal one; in particular, $Z_{n,1}=D$.
The remaining factor in the eigenvalues is $|O_k^{\rm st}(2)|$, the order
of the group of invertible binary linear maps preserving Hamming weight
modulo four and fixing the all-ones vector.  This order depends only on
$k$ and is given explicitly in Eq.~\eqref{eqS:stochastic-order}.
For every integer $r\geq0$, define
\begin{equation}
 F_r(x):=\frac{2^{r(r-1)}}{\prod_{j=1}^{r}(2^j-1)}
          \prod_{j=0}^{r-1}(x+2^{1-j}),
 \qquad F_0(x)=1.
 \label{eq:Fr}
\end{equation}

The spectral sectors have two branches, denoted by $\Ord$ and $\Exc$.
Their index $r$ is the dimension of an isotropic defect: an ordinary
defect excludes the all-ones vector, whereas an exceptional defect
contains it.  The precise defect construction is given in
Sections~\ref{sec:relations}--\ref{sec:filtration}.  Exceptional sectors
occur only when $4\mid k$.  We write $\lambda^b_{k,r}(n)$ and
$M^b_{k,r}(D)$ for the eigenvalue and its sector multiplicity in branch
$b\in\{\Ord,\Exc\}$; these scalar multiplicities are distinct from the
POVM operators $M_x$ in Eq.~\eqref{eq:operational}. Given the above notations, we state the exact optimal identification probability and give its proof sketch as follows.

\begin{theorem}[Exact finite-copy success probability and spectrum]
\label{thmS:finite-spectrum}
For every pair of integers $n,k\geq1$, the positive spectrum of $H_{n,k}$
consists of the following sectors:
\begin{enumerate}
\item The ordinary eigenvalues are
      \begin{equation}
       \lambda^{\Ord}_{k,r}(n)
       =\frac{N_n|O_k^{\rm st}(2)|}{Z_{n,k}}F_r(2^{n+3-k}),
       \label{eq:main-ordinary-eigenvalue}
      \end{equation}
      with multiplicities $M^{\Ord}_{k,r}(D)$ given by
      Eqs.~\eqref{eqS:M-odd}, \eqref{eqS:M-two}, and \eqref{eqS:M-zero}
      for $k$ odd, $k\equiv2\pmod4$, and $4\mid k$, respectively.
      The admissible values of $r$ are listed in Eq.~\eqref{eqS:ranges},
      with the endpoint conventions in Section~\ref{sec:multiplicity}.
\item When $4\mid k$, the exceptional eigenvalues and multiplicities are
      \begin{equation}
       \begin{aligned}
        \lambda^{\Exc}_{k,r}(n)&=\lambda^{\Ord}_{k-1,r-1}(n),\\
        M^{\Exc}_{k,r}(D)&=D^{-1}M^{\Ord}_{k-1,r-1}(D),
       \end{aligned}
       \label{eq:main-exceptional-spectrum}
      \end{equation}
      where $r-1$ ranges over the admissible ordinary defect indices at $k-1$,
      equivalently $r$ satisfies Eq.~\eqref{eqS:exceptional-ranges}.
\end{enumerate}
Sectors with zero multiplicity at $D=2^n$ are omitted.  The spectral
multiplicity of a repeated eigenvalue is the sum of its sector
multiplicities.  Let $a=(b,r)$ run over the
remaining sectors, with $\lambda_a(n,k):=\lambda^b_{k,r}(n)$ and
$M_a(n,k):=M^b_{k,r}(2^n)$.  All remaining Gram eigenvalues are zero,
with multiplicity $N_n-\sum_a M_a(n,k)$.
The pretty-good measurement is optimal, and
\begin{equation}
 P^*_{n,k}
 =\left[\frac1{N_n}\sum_aM_a(n,k)\sqrt{\lambda_a(n,k)}\right]^2,
 \label{eqS:finite-success}
\end{equation}
where $N_n$ is the ensemble size in Eq.~\eqref{eq:number}.
\end{theorem}

\begin{proof}[Proof sketch]
The Clifford group acts transitively on the equally likely pure states.
Section~\ref{sec:measurement} uses this symmetry to prove PGM optimality
and reduce $P^*_{n,k}$ to $(\Tr\sqrt{H_{n,k}}/N_n)^2$.
Thus it remains to determine the nonzero Gram spectrum.

Section~\ref{sec:relations} expresses the stabilizer moment as a sum of
binary relation operators divided by $Z_{n,k}$.
The nonzero eigenvalues of this relation sum become Gram eigenvalues
after multiplication by $N_n/Z_{n,k}$.
For $n\geq k$, Section~\ref{sec:filtration} decomposes the
stabilizer-tensor span into physical sectors indexed by defects.
Section~\ref{sec:ordinary} evaluates the relation sum on each ordinary
sector through weighted flag counts and a parabolic group count, giving
Eq.~\eqref{eq:main-ordinary-eigenvalue}.

When $4\mid k$, Section~\ref{sec:exceptional} shortens the all-ones code
to the ordinary problem at $k-1$.  The resulting operator identity gives
Eq.~\eqref{eq:main-exceptional-spectrum}.
Section~\ref{sec:multiplicity} then obtains the sector dimensions from
differences of invariant-code dimensions; shortening separates the
ordinary and exceptional contributions.

These steps give the complete spectrum in the stable range $n\geq k$.
Section~\ref{sec:finite} extends it to every physical value $D=2^n$
using polynomial identities for the power sums, together with positivity
of the surviving eigenvalues and nonnegative integer multiplicities.
Subtracting the resulting rank from $N_n$ gives the Gram nullity.
The trace formula now yields Eq.~\eqref{eqS:finite-success}; the complete
proof is assembled in Appendix~\ref{app:proof-finite-spectrum}.
\end{proof}

An important consequence of Theorem~\ref{thmS:finite-spectrum} is the
following exact curve in the critical window.  Fix $c\in\mathbb Z$ and
put $k=n+c$.  The theorem applies also when $c>0$, where $k>n$ lies outside
the stable decomposition range.  Section~\ref{sec:limit} derives the limits
of the ordinary sector weights $M^{\Ord}_{k,r}(2^n)/N_n$ and eigenvalues
$\lambda^{\Ord}_{k,r}(n)$.  To pass from these sector limits to the full
spectral sum, it proves a summable bound $C_c2^{-r^2}$, with $C_c<\infty$
depending only on $c$, for the ordinary amplitude summands
$(M^{\Ord}_{k,r}(2^n)/N_n)\sqrt{\lambda^{\Ord}_{k,r}(n)}$.
The bound is uniform in all sufficiently large $n$ and all admissible $r$,
including terminal sectors and every residue class of $k$.
The exceptional contribution to $\sqrt{P^*_{n,k}}$ is at most $2^{-n}$
by Eq.~\eqref{eqS:exceptional-bound}.  Dominated convergence therefore
gives Theorem~\ref{thm:curve}.  Section~\ref{sec:limit} states the limiting
estimates and their consequences, with detailed proofs in
Appendix~\ref{app:limit}.  Section~\ref{sec:tails} gives the two tail
laws, whose derivations are in Appendix~\ref{app:tails}.

\begin{theorem}[Fixed-offset curve]
\label{thm:curve}
For every fixed integer $c$,
\begin{equation}
 \lim_{n\to\infty}P^*_{n,n+c}=P_\infty(c),
 P_\infty(c)=\left[\sum_{r=0}^{\infty}
 w_r(c)\sqrt{\lambda_r(c)}\right]^2 .
 \label{eq:curve}
\end{equation}
The absolute constants $A,Q>0$ are the convergent products
\begin{equation}
 A:=\prod_{j\geq1}(1+2^{-j}),
 Q:=\prod_{j\geq1}(1-4^{-j}),
 \label{eq:AQ}
\end{equation}
independent of $n$, $k$, $c$, and $r$.
For each integer $r\geq0$, put $u_r:=2r-c$ and define
\begin{align}
  w_r(c)&:=\frac{2^{-(u_r^2+3u_r+4)/2}(1+2^{-u_r-2})}{AQ},\label{eq:weight}\\
 \lambda_0(c)&:=2^{(c-1)(c-2)/2+1}QR_c,\label{eq:lambda0}\\
 \lambda_r(c)&:=\lambda_0(c)F_r(2^{3-c}),\label{eq:lambdar}
\end{align}
where $F_r$ is given in Eq.~\eqref{eq:Fr}, and the positive factor $R_c$
depends only on $c$:
\begin{equation}
 R_c:=
 \begin{cases}
 \displaystyle\prod_{j=1}^{1-c}(1+2^{-j}),&c\leq1,\\[1ex]
 \displaystyle\left[\prod_{j=2-c}^{0}(1+2^{-j})\right]^{-1},&c\geq2.
 \end{cases}
 \label{eq:Rc}
\end{equation}
In particular, $R_1=1$ by the empty-product convention.

\end{theorem}
Here $r$ labels the ordinary spectral sectors: $w_r(c)$ is the limit of
$M^{\Ord}_{k,r}(2^n)/N_n$, and $\lambda_r(c)$ is the limit of
$\lambda^{\Ord}_{k,r}(n)$ at fixed $c,r$.  These weights satisfy
$\sum_rw_r(c)\leq1$; the
remaining spectral mass is at zero.  The normalization proved in
Appendix~\ref{app:proof-curve} concerns the first moment and retains the Gram
matrix's unit mean eigenvalue:
\begin{equation}
 \sum_{r\geq0}w_r(c)\lambda_r(c)=1.
 \label{eq:normalization}
\end{equation} The series is positive and super-exponentially convergent in $r$.
For example, $P_\infty(0)\approx0.1761$ and
$P_\infty(2)\approx0.7026$: two additional copies yield an order-one
change in the limiting success probability.
Figure~\ref{fig:num-sampling} in Section~\ref{sec:numerics} compares
simulated success frequencies with the exact finite-copy probabilities
and the limiting curve, while Fig.~\ref{fig:num-log} resolves the
deterministic finite-size corrections on a logarithmic scale.
At $c=0$, Table~\ref{tab:num-center} shows that the deviation
$|P^*_{n,n}-P_\infty(0)|$ decreases from $7.25\times10^{-2}$ at $n=4$
to $3.64\times10^{-6}$ at $n=32$ and $7.01\times10^{-40}$ at $n=256$.
Moreover, as $m\to\infty$,
\begin{equation}
 P_\infty(-m)\sim \frac{2^{-(m^2+3m+4)/2}}{AQ},\,
 1-P_\infty(m)\sim2^{-m}.
\end{equation}
These tails show that the crossover is neither Gaussian nor logistic: the subcritical
side has a quadratic-exponential tail, whereas every additional copy above
threshold asymptotically halves the error.

\section{Optimal measurement and Gram reduction}
\label{sec:measurement}
We begin the proof of the main results by reducing the optimization over
measurements to a spectral calculation.  The transitive Clifford symmetry
of the uniform pure-state ensemble makes the PGM optimal and expresses
its success probability through the square root of the Gram matrix.
The Gram and frame operators have the same nonzero spectrum, so the
remaining calculation can be carried out using the stabilizer moment
operator introduced below.

\subsection{Gram and frame operators}
Set $D=2^n$ and use $N_n$ from Eq.~\eqref{eq:number}.
The projective
Clifford group acts transitively on $\Stab_n$.  Given $k\geq1$, define the
synthesis map
\begin{equation}
 \Phi_{n,k}:\mathbb C^{N_n}\longrightarrow
 (\mathbb C^D)^{\otimes k},
 \qquad
 \Phi_{n,k}e_x=|x\rangle^{\otimes k}.
 \label{eqS:analysis}
\end{equation}
Its Gram and frame operators are
\begin{equation}
 H_{n,k}:=\Phi_{n,k}^{\dagger}\Phi_{n,k},
 \qquad
 S_{n,k}:=\Phi_{n,k}\Phi_{n,k}^{\dagger}
          =N_n\rho_{n,k},
 \label{eqS:gram-frame}
\end{equation}
where
\begin{equation}
 \rho_{n,k}:=\frac1{N_n}\sum_{x\in\Stab_n}\rho_x^{\otimes k}.
 \label{eqS:moment}
\end{equation}
The polar decomposition of $\Phi_{n,k}$ shows that $H_{n,k}$ and
$S_{n,k}$ have the same nonzero eigenvalues, including multiplicities.

\subsection{Exact PGM optimality}
We record a general statement, including the case in which the ensemble is
linearly dependent.

\begin{lemma}[Transitive pure orbits]
\label{lemS:PGM}
Let $\{|\phi_x\rangle:x\in X\}$ be equiprobable pure states, $|X|=N$, and
suppose a unitary group acts transitively on their rays.  Then the
pretty-good measurement is minimum-error optimal and
\begin{equation}
 P^*=P_{\rm PGM}
 =\left(\frac{\Tr\sqrt H}{N}\right)^2,
 \qquad H_{xy}=\langle\phi_x|\phi_y\rangle.
 \label{eqS:PGMformula}
\end{equation}
\end{lemma}

\noindent\emph{Proof:} See Appendix~\ref{app:proof-pgm}.\par

The PGM also attains the optimal worst-case conditional success
probability for this transitive pure-state ensemble.  Thus the
uniform-prior optimum is also the unrestricted minimax value;
the argument is given in Appendix~\ref{app:minimax}.

\section{Binary stochastic relations}
\label{sec:relations}
The reduction in Section~\ref{sec:measurement} leaves the Gram eigenvalues
and their multiplicities to be determined.  We now express the stabilizer
moment as a sum of operators indexed by binary stochastic relations.
Their quadratic and defect data identify the ordinary and exceptional
branches, while the associated group orders provide the counting factors
used later.  This prepares the physical sector decomposition in
Section~\ref{sec:filtration}.

\subsection{Quadratic relations and the moment operator}
We follow
Ref.~\cite{GrossNezamiWalter2021}.  Let $E_k=\F_2^k$, let
$b(x,y)=x\cdot y\pmod2$, and retain the integer-valued quadratic refinement
\begin{equation}
 q_k(x):=\sum_{i=1}^{k}x_i\pmod4,
 \qquad x_i\in\{0,1\}.
 \label{eqS:quadratic}
\end{equation}
Its polar form is $b$, in the convention
$q_k(x+y)=q_k(x)+q_k(y)+2b(x,y)\pmod4$.  The all-ones vector $\one_k$
is the \emph{Wu vector}: it is characterized by
$b(\one_k,x)=b(x,x)$ for every $x$.  Below, isotropic means that the full
mod-four refinement vanishes.  The adjective \emph{stochastic} refers to
preservation of this all-ones vector.
On $E_k\oplus E_k$ use
$q_{k,k}(x,y)=q_k(x)-q_k(y)\pmod4$.  A \emph{stochastic Lagrangian
relation} is a $k$-dimensional subspace
$T\leq E_k\oplus E_k$ on which $q_{k,k}$ vanishes and which contains
$(\one,\one)$.  Denote the finite set of these relations by
$\Sigma_{k,k}(2)$.

Regroup $(\mathbb C^D)^{\otimes k}$ into $n$ copies of
$(\mathbb C^2)^{\otimes k}$.  For $T\in\Sigma_{k,k}(2)$ set
\begin{equation}
 r(T):=\sum_{(x,y)\in T}|x\rangle\!\langle y|,
 \qquad R_n(T):=r(T)^{\otimes n}.
 \label{eqS:Rrelation}
\end{equation}
The exact stabilizer moment is
\begin{equation}
 \rho_{n,k}=\frac1{Z_{n,k}}
 \sum_{T\in\Sigma_{k,k}(2)}R_n(T),
 \qquad
 Z_{n,k}=D\prod_{j=0}^{k-2}(D+2^j).
 \label{eqS:moment-GNW}
\end{equation}
This is Theorem 5.3 of Ref.~\cite{GrossNezamiWalter2021}; the formula remains
valid outside the stable commutant range.

For a relation $T$ define its left and right kernels
\begin{equation}
 N_L(T):=\{x:(x,0)\in T\},\qquad
 N_R(T):=\{y:(0,y)\in T\}.
 \label{eqS:defects}
\end{equation}
They are totally $q_k$-isotropic, have the same dimension, and $T$ induces
an isometry
$N_R(T)^\perp/N_R(T)\to N_L(T)^\perp/N_L(T)$.  Composition of linear
relations gives the semigroup law
\begin{equation}
 R_n(T_1)R_n(T_2)
 =D^{\dim[N_R(T_1)\cap N_L(T_2)]}R_n(T_1\circ T_2).
 \label{eqS:semigroup}
\end{equation}
The exponent counts closed binary summation loops.  Also
\begin{equation}
 \Tr R_n(T)=D^{\dim(T\cap\Delta)},
 \qquad \Delta:=\{(x,x):x\in E_k\}.
 \label{eqS:relation-trace}
\end{equation}

The defect-zero relations are graphs of the stochastic orthogonal group
\begin{equation}
 O_k^{\rm st}(2):=
 \{O\in\mathrm{GL}(E_k):q_k(Ox)=q_k(x),\ O\one=\one\}.
 \label{eqS:stochastic-group}
\end{equation}
For an isotropic defect $N$, write
$Q_N=N^\perp/N$ and $G_N=\St(Q_N)$ for the group preserving the induced
quadratic data and stochastic vector.  Ordinary defects have $\one\notin N$.
When $4\mid k$, $q_k(\one)=0$ and defects containing $\one$ form an
additional exceptional family.

\subsection{Group orders and quadratic types}
For later use, define
\begin{align}
 |O^\delta(2m,2)|
 &=2^{m(m-1)+1}(2^m-\delta)
   \prod_{j=1}^{m-1}(2^{2j}-1),
 \label{eqS:orthogonal-order}\\
 |\Sp(2m,2)|
 &=2^{m^2}\prod_{j=1}^{m}(2^{2j}-1).
 \label{eqS:symplectic-order}
\end{align}
The orthogonal-order formula is for $m\geq1$; at the zero-dimensional plus
endpoint we use $|O^+(0,2)|=1$, while the zero-dimensional minus tower is
absent.
The order of Eq.~\eqref{eqS:stochastic-group} is
\begin{equation}
 |O_k^{\rm st}(2)|=
 \begin{cases}
 |O^\delta(k-1,2)|,&k\text{ odd},\\
 2|\Sp(k-2,2)|,&k\equiv2\pmod4,\\
 2^{k-2}|O^\delta(k-2,2)|,&4\mid k,
 \end{cases}
 \label{eqS:stochastic-order}
\end{equation}
where
\begin{equation}
 \delta=(-1)^{(k^2-1)/8}\quad(k\text{ odd}),
 \qquad
 \delta=(-1)^{k/4}\quad(4\mid k).
 \label{eqS:type}
\end{equation}
Here $O^\delta(2m,2)$ preserves a nondegenerate binary quadratic form:
its maximal isotropic dimension is $m$ for $\delta=+$ and $m-1$ for
$\delta=-$.  The group $\Sp(2m,2)$ preserves a nondegenerate alternating
bilinear form.  A Witt basis separates isotropic pairs from the residual
quadratic space.  Separating the Wu vector in such a basis gives the
stochastic group orders above; they agree
with the stochastic quadratic-space classification in
Ref.~\cite{MontealegreGross2025}.

We use the binary Gaussian binomial coefficient
\begin{equation}
 \gbinom{a}{b}{2}:=
 \prod_{j=0}^{b-1}\frac{2^{a-j}-1}{2^{b-j}-1},
 \qquad 0\leq b\leq a,
 \label{eq:gaussian-binomial}
\end{equation}
which counts the $b$-dimensional subspaces of $\F_2^a$.
We set $\gbinom{a}{0}{2}=1$ and $\gbinom{a}{b}{2}=0$ when
$b<0$ or $b>a$, for nonnegative integer $a$.

\section{Physical code sectors}
\label{sec:filtration}
The relation formula specifies the moment operator; the next task is to
decompose the physical stabilizer-tensor span into sectors on which it
acts by a scalar.  In the stable range $n\geq k$, the CSS code filtration
identifies the ordinary and exceptional sectors and proves that they
exhaust this span.  This reduces the spectral calculation to the scalars
and sector dimensions evaluated in the following three sections.
Section~\ref{sec:finite} will extend the resulting spectral formulas
to every physical dimension.

\subsection{CSS codes and the defect filtration}
Let
\begin{equation}
 A_k(D):=\sum_{T\in\Sigma_{k,k}(2)}R_n(T).
 \label{eqS:A}
\end{equation}
Throughout, $T\mapsto R_n(T)$ acts on the physical space
$(\mathbb C^D)^{\otimes k}$; no faithful formal relation representation is
assumed.  Equations~\eqref{eqS:gram-frame}, \eqref{eqS:moment},
\eqref{eqS:moment-GNW}, and~\eqref{eqS:A} give the normalization ledger
\begin{equation*}
 \begin{aligned}
 H_{n,k}&=\Phi_{n,k}^{\dagger}\Phi_{n,k},&
 S_{n,k}&=\Phi_{n,k}\Phi_{n,k}^{\dagger},\\
 S_{n,k}&=N_n\rho_{n,k}
          =\frac{N_n}{Z_{n,k}}A_k(D),&
 H_{n,k}&\sim_{\rm nz}S_{n,k},
 \end{aligned}
\end{equation*}
where $\sim_{\rm nz}$ means equality of nonzero spectra, including
multiplicities.  Accordingly, an eigenvalue $\theta$ of $A_k(D)$ corresponds
to the Gram eigenvalue $\lambda=(N_n/Z_{n,k})\theta$.
For a totally $q_k$-isotropic $r$-space $N$ and a linear map
$f:\F_2^n\to Q_N$, choose a linear lift
$F:\F_2^n\to N^\perp$ and define the normalized CSS coset vector
\begin{equation}
 |N,[f]\rangle:=
 D^{-r/2}\sum_{Z\in\operatorname{Hom}(\F_2^n,N)}|F+Z\rangle .
 \label{eqS:CSS-coset}
\end{equation}
Here a binary $k\times n$ matrix labels a computational-basis state; each
column records the $k$ copy coordinates at one physical qubit.  The sum is
independent of the lift.  These uniform coset vectors form an orthonormal
basis of the Calderbank--Shor--Steane (CSS) stabilizer code $C_N$; denote its
projector by $P_N^{\rm CSS}$.  The
Clifford code filtration of Ref.~\cite{MontealegreGross2025} is indexed by
these subspaces.  Let
$\cG_{k,r}$ (respectively $\cG^0_{k,r}$) be the set of $r$-dimensional
defects not containing (respectively containing) $\one_k$, and let $C_N$ be
the code in Eq.~\eqref{eqS:CSS-coset}.  Put
\begin{equation}
 \begin{aligned}
 \mathcal F'_r&:=\operatorname{span}\{C_N:N\in\cG_{k,r}\}\\
 &=\operatorname{span}\{C_N:\dim N\geq r,\ \one\notin N\}.
 \end{aligned}
 \label{eqS:raw-filtration}
\end{equation}
The equality follows from $C_{N'}\subset C_N$ for $N\subset N'$ and the
ordinary-defect orbit and extension result, Lemma III.4 of
Ref.~\cite{MontealegreGross2025}.  If $4\nmid k$, set
$\mathcal F_r^{\Ord}=\mathcal F'_r$.  If $4\mid k$, let
$P_{\one}:=P_{\langle\one_k\rangle}^{\rm CSS}$ and set
\begin{equation}
 \mathcal F_r^{\Ord}:=(I-P_{\one})\mathcal F'_r,
 \qquad
 \mathcal F_r^{\Exc}:=
 \operatorname{span}\{C_N:N\in\cG^0_{k,r}\}.
 \label{eqS:branch-filtrations}
\end{equation}
For each ordinary $N$, the code projectors satisfy
$P_{\one}P_N^{\rm CSS}=P_{N+\langle\one_k\rangle}^{\rm CSS}$.
Every exceptional $(r+1)$-defect arises in this way, so
$P_{\one}\mathcal F'_r=\mathcal F_{r+1}^{\Exc}$.  Thus these are nested
filtrations.  The index shift is geometric: adjoining the quotient Wu line
to an ordinary $r$-defect $N$ gives the exceptional defect
$N+\langle\one_k\rangle$ of dimension $r+1$, while quotienting an
exceptional defect by its all-ones line lowers its dimension by one.  Thus
shortening will send ordinary grade $r-1$ at $k-1$ to exceptional grade $r$
at $k$.  To separate successive defect dimensions, we take the part of
each filtration space orthogonal to the next, deeper space.  These
subspaces are the associated grades:
\begin{equation}
 \mathcal K_{k,r}:=
 \mathcal F_r^{\Ord}\cap(\mathcal F_{r+1}^{\Ord})^\perp,
 \qquad
 \mathcal L_{k,r}:=
 \mathcal F_r^{\Exc}\cap(\mathcal F_{r+1}^{\Exc})^\perp.
 \label{eqS:associated-grade}
\end{equation}
For $N\in\cG_{k,r}$ define the local grade
$\mathcal K_N:=C_N\cap\mathcal K_{k,r}$, and define
$\mathcal L_N:=C_N\cap\mathcal L_{k,r}$ on the exceptional branch.  In the
stable range $n\geq k$ the code-filtration theorem gives the direct sums
\begin{equation}
 \mathcal K_{k,r}=\bigoplus_{N\in\cG_{k,r}}\mathcal K_N,
 \qquad
 \mathcal L_{k,r}=\bigoplus_{N\in\cG^0_{k,r}}\mathcal L_N.
 \label{eqS:local-grades}
\end{equation}
In this section's sector decomposition, a \emph{physical Clifford module}
means an actual Clifford-invariant subspace of
$\operatorname{span}\{|S\rangle^{\otimes k}:S\in\Stab_n\}$ inside the
physical tensor-power Hilbert space, not an abstract multiplicity space in a
formal relation representation.  Intermediate symbols $\eta_{Q,L}(\tau)$
denote represented local Clifford submodules of a code $C_N$; the proposition
below identifies precisely which induced fixed copies are physical in the
stated global sense and proves that they exhaust the stabilizer-tensor span.

\subsection{From filtration grades to the physical spectrum}
The invariant part of a code generally contains contributions from several
defect grades, so its dimension is not yet a sector multiplicity.  The next recursion
removes contributions from deeper defects by comparing two quotient
codes.  Its index $L$ labels a defect inside $Q_N$; the original defect
in $E_k$ remains $N$.

\begin{lemma}[Invariant orbit-tail recursion]
\label{lemS:orbit-tail}
Assume $n\geq k$, fix an ordinary $N\in\cG_{k,r}$, and put
$Q=Q_N$, $s=\dim Q=k-2r$, and
$a_s(D)=\dim C_N^{G_N}$.  Let $\eta_{Q,L}(\bm1)$ denote the represented local
Clifford module paired with the trivial stochastic type at a defect orbit
$[L]$ of $Q$.  When an ordinary isotropic line $\ell<Q$ exists, define
$a_{s-2}(D)$ as the invariant-code dimension for
$Q'=\ell^\perp/\ell$; this stays in the same quadratic tower by
Proposition II.13 of Ref.~\cite{MontealegreGross2025}.  Beyond a terminal
Witt index set $a_{s-2}(D)=0$.  Explicitly, the formal terminal values
are $a_{-1}=0$ for the odd plus tower, $a_1=0$ for the odd minus tower,
$a_0=0$ for $k\equiv2\pmod4$, and $a_0=0$ or $a_2=0$ for the
$4\mid k$ plus or minus tower, respectively.  These are
zero-by-nonexistence conventions used to close the recursion; they are not
dimensions assigned to existing invariant codes.  Then
\begin{equation}
 a_s(D)=
 \dim\eta_{Q,0}(\bm1)
 +\bm1_{\,4\mid k}\dim\eta_{Q,\langle[\one_k]_N\rangle}(\bm1)
 +a_{s-2}(D).
 \label{eqS:orbit-tail-recursion}
\end{equation}
The \emph{heads} are the modules indexed by the zero defect and, only
when $4\mid k$, the quotient all-ones line.  All remaining quotient-defect
modules form the \emph{tail}, of dimension $a_{s-2}$.  Thus the difference
$a_s-a_{s-2}$ retains exactly these heads.  Here $\bm1_{\,4\mid k}$ is the
corresponding zero--one indicator.
\end{lemma}

\noindent\emph{Proof:} See Appendix~\ref{app:proof-orbit-tail}.\par

\begin{proposition}[Filtration, physical sectors, and exhaustion]
\label{propS:filtration-bridge}
Let $G=O_k^{\rm st}(2)$ and assume $n\geq k$.  For an ordinary
$N\in\cG_{k,r}$, let $P_N=\operatorname{Stab}_G(N)$ and inflate
$\tau\in\operatorname{Irr}G_N$ along the quotient map $P_N\to G_N$.
Here $\operatorname{Irr}$ denotes irreducible representations,
$\operatorname{Ind}$ below denotes induction from $P_N$ to $G$, and a
superscript $G$ denotes the subspace fixed by every element of $G$.
Then there are pairwise inequivalent irreducible Clifford modules
$\eta^{\Ord}_{k,r}(\tau)$ such that
\begin{equation}
 \mathcal K_{k,r}\simeq
 \bigoplus_{\tau\in\operatorname{Irr}G_N}
 \operatorname{Ind}_{P_N}^{G}(\tau)\otimes
 \eta^{\Ord}_{k,r}(\tau).
 \label{eqS:induced-grade}
\end{equation}
The same statement holds for $\mathcal L_{k,r}$ with an exceptional
representative and modules $\eta^{\Exc}_{k,r}(\tau)$.  Consequently
\begin{equation}
 \mathcal K_{k,r}^{G}\simeq\eta^{\Ord}_{k,r}(\bm1),\qquad
 \mathcal L_{k,r}^{G}\simeq\eta^{\Exc}_{k,r}(\bm1).
 \label{eqS:fixed-grade}
\end{equation}
These are single physical Clifford modules, not multiplicity spaces in an
abstract relation representation.

Put $s=k-2r$ and
$a_s(D):=\dim C_N^{G_N}$, using the terminal-tower convention of
Lemma~\ref{lemS:orbit-tail}.  If $4\nmid k$, then
\begin{equation}
 \dim\mathcal K_{k,r}^{G}=a_s(D)-a_{s-2}(D).
 \label{eqS:adjacent-difference}
\end{equation}
If $4\mid k$, the adjacent drop splits without overlap as
\begin{equation}
 a_s(D)-a_{s-2}(D)
 =\dim\mathcal K_{k,r}^{G}
  +\dim\mathcal L_{k,r+1}^{G}.
 \label{eqS:ordinary-exceptional-split}
\end{equation}
Finally, the stabilizer-tensor span is the orthogonal, multiplicity-free sum
of all the modules in Eq.~\eqref{eqS:fixed-grade}.  On every other
$G$-type the relation sum $A_k(D)$ is zero.
\end{proposition}

\noindent\emph{Proof:} See Appendix~\ref{app:proof-filtration-bridge}.\par

For $n\geq k$, this proposition completes the reduction from code spaces
to scalar spectral sectors.  The fixed parts are fixed by the stochastic
group $G$, while their irreducibility refers to the Clifford action.
The relation sum $A_k(D)$ commutes with that action and vanishes outside
the stabilizer-tensor span.  Since these fixed parts exhaust this span as
pairwise inequivalent irreducible Clifford modules, Schur's lemma makes
$A_k(D)$ scalar on each of them.  We now compute those scalars.

\section{Ordinary eigenvalues}
\label{sec:ordinary}
With the physical sectors established in Section~\ref{sec:filtration},
we compute the scalar on each ordinary sector in the stable range
$n\geq k$.  Mixed-signature moments and physical CSS compression,
together with affine overlaps and weighted flag counts, give the master
scalar identity~\eqref{eqS:spherical-ratio}, derived in
Appendix~\ref{app:proof-flags}.  Orbit--stabilizer counting
and the parabolic kernel then reduce its factors to the product in
Eq.~\eqref{eqS:theta}.  Rescaling by $N_n/Z_{n,k}$ gives the ordinary
Gram eigenvalues in Eq.~\eqref{eqS:lambda-finite}; the exceptional branch
is treated next.

\subsection{Mixed-signature moments}
Under the CSS code isometry, the Clifford action restricted to $C_N$
becomes a tensor representation associated with $Q_N$.  This
representation can include complex-conjugate factors, so the ordinary
$s$-copy moment does not cover every quotient code.  The following
mixed-signature identity supplies the moment needed for this compression;
the two signs distinguish ordinary and complex-conjugate factors.

\begin{lemma}[Mixed-signature stabilizer moment]
\label{lemS:mixed-moment}
Let $V=A\oplus B\simeq\F_2^t$ carry the signed Hamming refinement
\begin{equation}
 q_{a,b}(x_A,x_B)=q_a(x_A)-q_b(x_B)\pmod4,
 \qquad a+b=t,
 \end{equation}
and let $w=\one_t$ be its Wu vector.  Define
\begin{equation}
 \begin{aligned}
 \Sigma(q_{a,b},q_{a,b}):=\{T\leq V\oplus V:\;&\dim T=t,\\
 (q_{a,b}\oplus-q_{a,b})|_T=0,\;&(w,w)\in T\},
 \end{aligned}
\end{equation}
and, for $T$ in this set,
\begin{equation}
 r_{q_{a,b}}(T):=\sum_{(x,y)\in T}|x\rangle\!\langle y|,
 \qquad R_{n,q_{a,b}}(T):=r_{q_{a,b}}(T)^{\otimes n}.
\end{equation}
In the mixed tensor representation
$\Delta_{a,b}(U)=U^{\otimes a}\otimes\overline U^{\otimes b}$, put
$|S;q_{a,b}\rangle:=|S\rangle^{\otimes a}\otimes
|\overline S\rangle^{\otimes b}
=\Delta_{a,b}(U_S)|0^n\rangle^{\otimes t}$ whenever
$U_S|0^n\rangle$ represents the projective stabilizer state $S$.  Then the
rank-one projector is independent of $U_S$ and, for every $n,t\geq1$,
\begin{equation}
 \begin{aligned}
 &\frac1{N_n}\sum_{S\in\Stab_n}|S;q_{a,b}\rangle\!\langle S;q_{a,b}|\\
 &\qquad=\frac1{Z_{n,t}}
   \sum_{T\in\Sigma(q_{a,b},q_{a,b})}R_{n,q_{a,b}}(T),\\
 &Z_{n,t}=D\prod_{j=0}^{t-2}(D+2^j).
 \end{aligned}
 \label{eqS:mixed-moment-lemma}
\end{equation}
Moreover,
\begin{equation}
 \begin{gathered}
 \bigl|\{T:\dim(T\cap\Delta_V)=t-\ell\}\bigr|
 =2^{\binom\ell2}\gbinom{t-1}{\ell}{2},\\
 0\leq\ell\leq t-1.
 \end{gathered}
 \label{eqS:mixed-diagonal-count}
\end{equation}
The same statements hold after transporting the quadratic space by an
isometry.
\end{lemma}

\noindent\emph{Proof:} See Appendix~\ref{app:proof-mixed-moment}.\par

\subsection{Affine CSS overlaps and weighted flags}
CSS codes associated with different defects can overlap.  To extract a
sector eigenvalue from a compressed matrix element, we need the weight
of a CSS vector's projection onto the stochastic fixed subspace.  The
following overlap calculation supplies that weight by group averaging.

\begin{lemma}[Affine CSS overlap and weighted defect flags]
\label{lemS:weighted-flags}
Fix an ordinary isotropic $r$-space $N$, set $G=O_k^{\rm st}(2)$, choose a
surjection $F:\F_2^n\to N^\perp$, and take $O\in G$.  Put
$M=ON$, $K=N\cap M$, $\ell=r-\dim K$, and let
$\kappa_N=|\ker(\operatorname{Stab}_G(N)\to G_N)|$.
Consider the CSS vectors supported on
$F+\operatorname{Hom}(\F_2^n,N)$ and
$OF+\operatorname{Hom}(\F_2^n,M)$, respectively, each with amplitude
$D^{-r/2}$ on its support.  The second vector is $R_n(O)$ applied to the
first.  Their supports meet if and only if
\begin{equation}
 (O-I)N^\perp\subseteq N+M,
 \label{eqS:weighted-overlap-condition}
\end{equation}
and the inner product of the two normalized vectors is then $D^{-\ell}$.

Condition~\eqref{eqS:weighted-overlap-condition} holds precisely when the
pairing $N/K\times M/K\to\F_2$ is nondegenerate and the map induced by $O$
on quotient spaces is the canonical stochastic isometry
\begin{equation}
 J_{M,N}:Q_N\longrightarrow Q_M,
 \qquad [x]_N\longmapsto[x+n_x]_M,
 \label{eqS:weighted-canonical-map}
\end{equation}
where $n_x\in N$ is the unique class modulo $K$ for which
$x+n_x\in M^\perp$.  For every such $M$---that is, every ordinary
$r$-defect with $N\cap M=K$ and nondegenerate pairing
$N/K\times M/K$---exactly $\kappa_N$ elements $O$ induce this map.
Finally, for a fixed $K<N$ of codimension $\ell$,
the number of such $M$ is
\begin{equation}
 2^{\ell(s-1)+\binom\ell2},
 \qquad s=k-2r.
 \label{eqS:weighted-flag-count}
\end{equation}
\end{lemma}

\noindent\emph{Proof:} See Appendix~\ref{app:proof-weighted-flags}.\par

\subsection{The scalar from physical compression}
For the CSS vector specified above, compression through the mixed moment
gives a diagonal matrix element of $A_k(D)$.  The overlap average gives
its projection weight onto the stochastic fixed part of the ordinary
grade.  On that fixed part $A_k(D)$ is scalar, and it vanishes on the
other stochastic types.  Dividing the matrix element by the projection
weight therefore gives the scalar below, as derived in
Appendix~\ref{app:proof-flags}.

\begin{lemma}[Spherical CSS normalization]
\label{lemS:flags}
Fix an ordinary totally isotropic $r$-space $N$, put $s=k-2r$, and work
first in the stable range $n\geq k$.  On the
$O_k^{\rm st}(2)$-fixed sector of the ordinary defect-$r$ grade, the relation
sum has eigenvalue
\begin{equation}
 |\cG_{k,r}|\,|G_N|\,H_{k,r}(D),
 \qquad
 H_{k,r}(D)=\prod_{\ell=k-r-1}^{k-2}(D+2^\ell).
 \label{eqS:H}
\end{equation}
\end{lemma}

\noindent\emph{Proof:} See Appendix~\ref{app:proof-flags}.\par

The scalar includes both the CSS compression and the overlap
normalization.  Appendix~\ref{app:normalization-example} illustrates
their separate contributions at $k=5,r=1$.

\subsection{The parabolic kernel and ordinary eigenvalues}
\begin{lemma}[Ordinary parabolic kernel]
\label{lemS:parabolic}
Let $\cG_{k,r}$ be an ordinary $O_k^{\rm st}(2)$-orbit of isotropic
$r$-defects and let $N\in\cG_{k,r}$.  If
$P_r=\operatorname{Stab}_{O_k^{\rm st}(2)}(N)$, the natural map
$P_r\to G_N$ is onto and
\begin{equation}
 |\ker(P_r\to G_N)|
 =2^{r(k-r-2)}\prod_{j=1}^{r}(2^j-1).
 \label{eqS:kernel}
\end{equation}
Equivalently,
\begin{equation}
 |\cG_{k,r}|\,|G_N|
 =|O_k^{\rm st}(2)|
 \frac{2^{r(r+2-k)}}{\prod_{j=1}^{r}(2^j-1)}.
 \label{eqS:orbit-product}
\end{equation}
\end{lemma}

\noindent\emph{Proof:} See Appendix~\ref{app:proof-parabolic}.\par

We use the polynomial $F_r$ defined in Eq.~\eqref{eq:Fr} to collect the
product factors in the ordinary eigenvalues.
\begin{proposition}[Ordinary eigenvalues]
\label{propS:eigenvalues}
Assume $n\geq k$, let $r$ be an admissible ordinary defect, and use
$F_r$ from Eq.~\eqref{eq:Fr}, with $F_0=1$.
The nonzero eigenvalue of $A_k(D)$ on the
$O_k^{\rm st}(2)$-fixed sector of its ordinary defect-$r$ grade is
\begin{equation}
 \theta_{k,r}(D)=|O_k^{\rm st}(2)|F_r(2^{n+3-k}).
 \label{eqS:theta}
\end{equation}
The corresponding Gram eigenvalue, i.e., the eigenvalue of $H_{n,k}$, is
\begin{equation}
 \lambda^{\Ord}_{k,r}(n)
 =\frac{N_n|O_k^{\rm st}(2)|}{Z_{n,k}}
 F_r(2^{n+3-k}).
 \label{eqS:lambda-finite}
\end{equation}
\end{proposition}
The displayed formulas will also define candidate spectral data outside
this range; Theorem~\ref{thmS:finite-spectrum} establishes their complete
physical spectrum interpretation for every $n,k\geq1$.

\noindent\emph{Proof:} See Appendix~\ref{app:proof-eigenvalues}.\par

\section{Exceptional all-ones sectors}
\label{sec:exceptional}
The ordinary eigenvalues leave one additional branch when $4\mid k$:
the sectors whose defects contain the all-ones vector.  We determine
this branch by shortening to the ordinary problem at copy number $k-1$.
Each exceptional sector at $k$ inherits the Gram eigenvalue of its
ordinary counterpart at $k-1$ and has $1/D$ of its dimension.
The resulting operator identity therefore
supplies both the exceptional eigenvalues and the relation needed to
separate the two branches in Section~\ref{sec:multiplicity}.

Suppose $4\mid k$ and $N$ is an isotropic defect containing $\one_k$.
Shorten its last coordinate:
\begin{equation}
 N^-:=\{x\in\F_2^{k-1}:(x,0)\in N\}.
 \label{eqS:shorten}
\end{equation}
Then $\dim N^-=\dim N-1$, and $N^-$ is an ordinary defect for the odd
tensor power $k-1$.

\begin{lemma}[Operator shortening on the all-ones code]
\label{lemS:shortening}
Put $u=\one_k=(w,1)$ with $w=\one_{k-1}$, and write $P_u$ for the
$n$-column CSS projector of $\langle u\rangle$.  There is a unitary
\begin{equation}
 \mathcal U_n:(\mathbb C^D)^{\otimes(k-1)}
 \longrightarrow\mathbb C^D\otimes C_{\langle u\rangle}
 \label{eqS:shortening-unitary}
\end{equation}
such that the following identity holds on the entire all-ones code:
\begin{equation}
 I_D\otimes\bigl(P_uA_k(D)P_u\bigr)|_{C_{\langle u\rangle}}
 =(D+2^{k-2})\,\mathcal U_nA_{k-1}(D)\mathcal U_n^\dagger.
 \label{eqS:shortening-intertwiner}
\end{equation}
It maps the ordinary filtration at $k-1$ to the exceptional filtration at
$k$, with the defect index increased by one.  Consequently, on every
exceptional grade,
\begin{equation}
 \begin{aligned}
 M^{\Exc}_{k,r}(D)&=D^{-1}M^{\Ord}_{k-1,r-1}(D),\\
 \lambda^{\Exc}_{k,r}(n)&=\lambda^{\Ord}_{k-1,r-1}(n).
 \end{aligned}
 \label{eqS:shortening-spectrum}
\end{equation}
\end{lemma}

\noindent\emph{Proof:} See Appendix~\ref{app:proof-shortening}.\par

The factor $D+2^{k-2}$ in the relation-sum identity equals
$Z_{n,k}/Z_{n,k-1}$, so it cancels when the two sides are converted to
frame operators.  The tensor factor $I_D$ repeats each exceptional-sector
eigenvalue $D$ times, which accounts for the division by $D$ in its
multiplicity.  The quadratic-refinement and endpoint checks accompanying
this construction are collected in Appendix~\ref{app:quadratic-checks}.

\section{Eigenvalue multiplicities}
\label{sec:multiplicity}
Sections~\ref{sec:ordinary} and~\ref{sec:exceptional} determine the
eigenvalues and relate the exceptional branch to the shortened ordinary
problem.  To evaluate the success probability, we still need the
dimensions of the physical sectors.  Burnside averaging gives the
invariant-code dimensions, and adjacent differences isolate the grades;
when $4\mid k$, shortening separates their ordinary and exceptional
contributions.  The resulting multiplicity formulas and endpoint
conventions complete the stable-range spectral data needed for the
continuation in Section~\ref{sec:finite}.

\subsection{Fixed-space enumerators and invariant-code dimensions}
For a finite classical group $G$ acting on its natural binary module, define
the fixed-space enumerator
\begin{equation}
 B_G(z):=\frac1{|G|}\sum_{g\in G}z^{\dim\Fix(g)}.
 \label{eqS:fixed-enumerator}
\end{equation}
We abbreviate the enumerators of $O^\delta(2m,2)$ and $\Sp(2m,2)$ by
$B_m^{O,\delta}$ and $B_m^{\Sp}$, respectively.  For $m,\ell\geq1$,
exact fixed-space distributions in Ref.~\cite{FulmanStanton2016},
followed by the finite $q$-binomial identity, give
\begin{align}
 &B_m^{O,\delta}(z)-B_{m-1}^{O,\delta}(z)\notag\\
 &\quad=\frac{\displaystyle
 \prod_{j=0}^{m-2}(z^2-4^j)
 (z-\delta2^{m-1})(z+2^{2m-1})}
 {|O^\delta(2m,2)|},
 \label{eqS:Oincrement}\\
 &B_\ell^{\Sp}(z)-B_{\ell-1}^{\Sp}(z)\notag\\
 &\quad=\frac{\displaystyle
 (z-1)\prod_{j=1}^{\ell-1}(z^2-4^j)(z+2^{2\ell})}
 {|\Sp(2\ell,2)|}.
 \label{eqS:Spincrement}
\end{align}
Here empty products equal one, $B_0^{O,+}=B_0^{\Sp}=1$, and
$B_0^{O,-}=0$ is the formal initial value for the minus tower.  At the
$m=1$ endpoints we also set
$B_{-1}^{\Sp}=B_{-1}^{O,+}=0$.  For
completeness, Eqs.~\eqref{eqS:Oincrement} and \eqref{eqS:Spincrement} can be
checked directly by inserting Theorems 2.3 and 2.5 of
Ref.~\cite{FulmanStanton2016} and applying the finite $q$-binomial identity
in their Eq.~(3).  These are exact polynomial identities in $z$, not
limiting fixed-vector formulas.

Burnside averaging converts these fixed-space enumerators into the
invariant-code dimensions stated next.  The underlying permutation
representation and the count are given in Appendix~\ref{app:proof-invariant-code}.

\begin{lemma}[Invariant-code dimensions]
\label{lemS:invariant-code}
Put $s=\dim Q_N=k-2r$.  The invariant dimension $a_s(D)$ introduced in
Proposition~\ref{propS:filtration-bridge} is as follows.
For odd $k$,
\begin{equation}
 a_{2m+1}(D)=D B_m^{O,\delta}(D).
 \label{eqS:a-odd}
\end{equation}
For $k\equiv2\pmod4$,
\begin{equation}
 a_{2m}(D)=\frac{D(D+1)}2B_{m-1}^{\Sp}(D).
 \label{eqS:a-two}
\end{equation}
For $4\mid k$,
\begin{equation}
 \begin{aligned}
 a_{2m}(D)&=B_{m-1}^{O,\delta}(D)\\
 &\quad+(D^2-1)B_{m-1}^{O,\delta}(D/2).
 \end{aligned}
 \label{eqS:a-zero}
\end{equation}
\end{lemma}

\noindent\emph{Proof:} See Appendix~\ref{app:proof-invariant-code}.\par

The adjacent differences in Proposition~\ref{propS:filtration-bridge}
convert invariant-code dimensions into physical multiplicities,
with shortening separating the two branches when $4\mid k$.

\subsection{Explicit multiplicities and admissible ranges}
\begin{proposition}[Ordinary multiplicities]
\label{propS:multiplicity}
Assume $n\geq k$ and let $s=k-2r$.  The multiplicity $M^{\Ord}_{k,r}(D)$ of the eigenvalue
\eqref{eqS:lambda-finite} is the following explicit integer.

If $k$ is odd, write $s=2m+1$ and use the sign in
Eq.~\eqref{eqS:type}.  Then
\begin{equation}
 M^{\Ord}_{k,r}(D)=
 \frac{D(D-\delta2^{m-1})(D+2^{2m-1})
 \prod_{j=0}^{m-2}(D^2-4^j)}{|O^\delta(2m,2)|}.
 \label{eqS:M-odd}
\end{equation}
For $m=0$ in the plus tower, Eq.~\eqref{eqS:M-odd} is interpreted as
$M=D$.

If $k\equiv2\pmod4$, write $s=2m$.  Then
\begin{equation}
 M^{\Ord}_{k,r}(D)=
 \frac{D(D+2^{2m-2})
 \prod_{j=0}^{m-2}(D^2-4^j)}
 {2|\Sp(2m-2,2)|}.
 \label{eqS:M-two}
\end{equation}

If $4\mid k$, write $s=2m$ and use the sign in
Eq.~\eqref{eqS:type}.  The ordinary multiplicity is
\begin{equation}
 M^{\Ord}_{k,r}(D)=
 \frac{(D-\delta2^{m-1})(D+2^{2m-2})
 \prod_{j=0}^{m-2}(D^2-4^j)}
 {2^{2m-2}|O^\delta(2m-2,2)|}.
 \label{eqS:M-zero}
\end{equation}
\end{proposition}

\noindent\emph{Proof:} See Appendix~\ref{app:proof-multiplicity}.\par

The admissible ordinary ranges are useful when evaluating the formulas:
\begin{equation}
 \begin{array}{c|c}
 \text{case}&\text{admissible }r\\ \hline
 k\text{ odd},\ \delta=+&0\leq r\leq(k-1)/2\\
 k\text{ odd},\ \delta=-&0\leq r\leq(k-3)/2\\
 k\equiv2\pmod4&0\leq r\leq k/2-1\\
 4\mid k,\ \delta=+&0\leq r\leq k/2-1\\
 4\mid k,\ \delta=-&0\leq r\leq k/2-2.
 \end{array}
 \label{eqS:ranges}
\end{equation}
At a physical value $D=2^n$, candidates with sufficiently large residual
quotient dimension $s=k-2r$ vanish: eventually the product contains the
factor $D^2-4^n=0$, but the first candidates outside the stable range can
survive.  Nonnegative integrality of every
displayed multiplicity at every physical specialization is proved
independently in Lemma~\ref{lemS:polynomial-continuation} below.
For reference, the exceptional ranges obtained by shortening are
\begin{equation}
 \begin{array}{c|c}
 k\equiv0\pmod8&1\leq r\leq k/2\\
 k\equiv4\pmod8&1\leq r\leq k/2-1.
 \end{array}
 \label{eqS:exceptional-ranges}
\end{equation}

\section{Complete spectrum at every physical dimension}
\label{sec:finite}
Sections~\ref{sec:filtration}--\ref{sec:multiplicity} determine the
spectrum using the stable code decomposition for $n\geq k$.  The
positive-offset window has $k>n$, so its analysis requires extending
these formulas beyond that range.  We combine polynomial identities
for the power sums with positivity of the surviving eigenvalues and
nonnegative integer multiplicities to establish the spectrum at every
physical $n,k\geq1$.  Counting the zero eigenspaces separately completes
Theorem~\ref{thmS:finite-spectrum} and supplies the finite probabilities
needed for the fixed-offset limit.

\subsection{Continuation beyond the stable range}
For fixed $k$, every positive-integer power sum of the actual relation
operator is a polynomial in $D$.  The power sums of the candidate
eigenvalues, weighted by their multiplicities, are also polynomials.
They agree at infinitely many stable values $D=2^n$, and hence agree
identically.  At each physical specialization, the proof independently
establishes nonnegative integer multiplicities and positive surviving
eigenvalues; equality of all positive moments then identifies the
nonzero spectrum.  This conclusion concerns the spectral multiset,
while the physical-module decomposition remains restricted to $n\geq k$.

\begin{lemma}[Polynomial continuation of a positive finite spectrum]
\label{lemS:polynomial-continuation}
Fix $k$ and first declare the complete finite candidate set
\begin{equation*}
 \begin{aligned}
 \mathfrak A_k^{\Ord}
  &:=\{(\Ord,r):0\leq r\leq\lfloor k/2\rfloor\},\\
 \mathfrak A_k^{\Exc}
  &:=\begin{cases}
      \{(\Exc,r):1\leq r\leq k/2\},&4\mid k,\\
      \varnothing,&4\nmid k,
     \end{cases}\\
 \mathfrak A_k&:=\mathfrak A_k^{\Ord}\cup\mathfrak A_k^{\Exc}.
 \end{aligned}
\end{equation*}
For an ordinary label in Eq.~\eqref{eqS:ranges}, define $M_a(D)$ by
Eqs.~\eqref{eqS:M-odd}--\eqref{eqS:M-zero}\noeqref{eqS:M-two} and $\theta_a(D)$ by
Eq.~\eqref{eqS:theta}.  For an exceptional label in
Eq.~\eqref{eqS:exceptional-ranges}, use
Eq.~\eqref{eqS:shortening-spectrum}; at the level of $A_k(D)$ this means
$\theta^{\Exc}_{k,r}=(D+2^{k-2})\theta^{\Ord}_{k-1,r-1}$.  Every formally
absent endpoint label in this displayed candidate set is assigned
$M_a(D)=0$ and the same formal $\theta_a(D)$ formula.  Labels whose
multiplicity vanishes only after a physical specialization remain in
$\mathfrak A_k$ until that specialization.  If these data give the positive
spectrum for every stable specialization $D=2^n$, $n\geq k$, then they give
the positive spectrum at every physical specialization $D=2^n$ after
zero-multiplicity labels are discarded.  At such a specialization,
\begin{align}
 \operatorname{mult}_0 A_k(D)&=D^k-\sum_{a\in\mathfrak A_k}M_a(D),
 \label{eqS:A-zero-multiplicity}\\
 \operatorname{mult}_0 H_{n,k}&=N_n-\sum_{a\in\mathfrak A_k}M_a(D).
 \label{eqS:H-zero-multiplicity}
\end{align}
\end{lemma}

\noindent\emph{Proof:} See Appendix~\ref{app:proof-continuation}.\par

\subsection{Finite-copy optimum and spectral consequences}

\noindent\emph{Proof of Theorem~\ref{thmS:finite-spectrum}:} See Appendix~\ref{app:proof-finite-spectrum}.\par

As consequences of the theorem, the zeroth positive mass is at most the
appropriate ambient dimension, by
Eqs.~\eqref{eqS:A-zero-multiplicity}--\eqref{eqS:H-zero-multiplicity}, and
the first Gram moment is exactly
\begin{equation}
 \sum_aM_a(n,k)\lambda_a(n,k)=N_n.
 \label{eqS:first-moment}
\end{equation}
Equation~\eqref{eqS:first-moment} is also obtained immediately from
$\Tr H_{n,k}=N_n$.

Applied to $|x\rangle^{\otimes k}$, Lemma~\ref{lemS:PGM} gives
\begin{equation}
 P^*_{n,k}=\left(\frac1{N_n}\Tr\sqrt{H_{n,k}}\right)^2.
 \label{eqS:P-spectrum}
\end{equation}

When $4\mid k$, Theorem~\ref{thmS:finite-spectrum} says that the complete
positive spectrum at the odd copy number $k-1$ is ordinary.
Lemma~\ref{lemS:shortening} and Eq.~\eqref{eqS:P-spectrum} therefore give
the exact amplitude identity
\begin{equation}
 \sum_r\frac{M^{\Exc}_{k,r}}{N_n}
 \sqrt{\lambda^{\Exc}_{k,r}}
 =D^{-1}\sqrt{P^*_{n,k-1}}\leq D^{-1}=2^{-n}.
 \label{eqS:exceptional-bound}
\end{equation}

\section{Fixed-offset limit}
\label{sec:limit}
Theorem~\ref{thmS:finite-spectrum}, now proved for every $n,k\geq1$,
supplies the finite spectrum needed for the critical window.  Fix
$c\in\mathbb Z$ and set $k=n+c$.  As $n$ grows, the number of sectors
also grows, so fixed-sector limits alone do not determine the full
spectral sum.  We derive these limits together with a summable bound
that is uniform for all sufficiently large $n$, including all residue
classes and the moving terminal layer.  Combined with the exceptional-sector bound from
Section~\ref{sec:finite}, these estimates justify passing the limit
through the sum and prove Theorem~\ref{thm:curve}.
We use $A,Q,R_c,w_r(c),\lambda_r(c)$ from
Eqs.~\eqref{eq:AQ}, \eqref{eq:weight}, \eqref{eq:lambda0},
\eqref{eq:lambdar}, and~\eqref{eq:Rc}.

The finite inputs are the ordinary multiplicities $M^{\Ord}_{k,r}(D)$
and Gram eigenvalues $\lambda^{\Ord}_{k,r}(n)$ from
Sections~\ref{sec:multiplicity} and~\ref{sec:ordinary}, with $D=2^n$.
Write $s=k-2r$ for the residual dimension, so $s=2m+1$ for odd $k$
and $s=2m$ for even $k$.  Section~\ref{sec:relations} defines the
stochastic orthogonal group $O_k^{\rm st}(2)$, its type sign $\delta$
in Eq.~\eqref{eqS:type}, and the normalization
$Z_{n,k}=D\prod_{j=0}^{k-2}(D+2^j)$ in
Eq.~\eqref{eqS:moment-GNW}.  The admissible sectors and terminal
conventions are in Section~\ref{sec:multiplicity}.  We also use the
polynomial shorthand $F_r$ in Eq.~\eqref{eq:Fr};
Section~\ref{sec:finite} supplies the separate exceptional-amplitude
bound~\eqref{eqS:exceptional-bound}.

\subsection{Fixed-sector limits and uniform tails}
\begin{lemma}[Termwise limits and uniform tails]
\label{lemS:limit-bounds}
For every fixed $c,r\in\mathbb Z$ with $r\geq0$, along every residue class
of $k=n+c$,
\begin{equation}
 \frac{M^{\Ord}_{k,r}(2^n)}{N_n}\longrightarrow w_r(c),
 \qquad
 \lambda^{\Ord}_{k,r}(n)\longrightarrow\lambda_r(c).
 \label{eqS:termwise-limit}
\end{equation}
Moreover, there is a finite constant $C_c$ such that, uniformly in all
sufficiently large $n$ and all admissible $r\geq0$,
\begin{align}
 \frac{M^{\Ord}_{k,r}}{N_n}\sqrt{\lambda^{\Ord}_{k,r}}
 &\leq C_c2^{-r^2},
 \label{eqS:amplitude-majorant}\\
 \frac{M^{\Ord}_{k,r}}{N_n}\lambda^{\Ord}_{k,r}
 &\leq C_c2^{-r^2}.
 \label{eqS:moment-majorant}
\end{align}
\end{lemma}

\noindent\emph{Proof:} See Appendix~\ref{app:proof-limit-bounds}.\par

The amplitude bound controls the ordinary contribution to
$\sqrt{P^*_{n,k}}$; the moment bound controls the ordinary contribution
to the first spectral moment.  They permit
separate applications of dominated convergence to these two sums.

\subsection{The limiting curve and its normalization}
The exceptional first moment is $D^{-1}$ when $4\mid k$ and zero
otherwise, so the ordinary first-moment sum tends to one.  The moment
bound passes this normalization to the limit, giving
$\sum_rw_r(c)\lambda_r(c)=1$ as in Eq.~\eqref{eq:normalization}.
Together with the amplitude limit and the vanishing exceptional
amplitude, this proves the limiting curve and its normalization.

\noindent\emph{Proof of Theorem~\ref{thm:curve}:} See Appendix~\ref{app:proof-curve}.\par

The sum of the weights themselves is the total nonzero limiting
spectral mass and satisfies $\sum_r w_r(c)\leq1$; the remaining mass lies
at the zero eigenvalue.  At $c=0$, the positive spectral mass is
$0.1954251444668857\ldots$.  Appendix~\ref{app:spectral-mass}
gives the convergence and mass calculation.

\section{Tail asymptotics}
\label{sec:tails}
Section~\ref{sec:limit} determines $P_\infty(c)$ by taking $n\to\infty$
at each fixed integer $c$.  We now study the tails of this limiting
curve as $c\to-\infty$ and $c\to+\infty$.  The goal is to obtain the
leading decay laws and their exact constants, quantifying the
suppression of success below threshold and of failure above it.
The left tail is governed by the $r=0$ spectral amplitude.  On the right,
the exact first-moment identity removes the linear contribution at
leading order in the expansion of the square root, and the second-order
term determines the leading error.  The proof controls the remaining
spectral sectors before applying this expansion.

\begin{proposition}[Two tails]
\label{propS:tails}
As $m\to\infty$ through the integers,
\begin{align}
 P_\infty(-m)
 &=\frac{2^{-(m^2+3m+4)/2}}{AQ}
   [1+o(1)],
 \label{eqS:left-asymptotic}\\
 1-P_\infty(m)
 &=2^{-m}[1+o(1)].
 \label{eqS:right-asymptotic}
\end{align}
\end{proposition}

\noindent\emph{Proof:} See Appendix~\ref{app:proof-tails}.\par

\section{Numerical study of the critical window}
\label{sec:numerics}
Sections~\ref{sec:finite}--\ref{sec:tails} establish the exact finite
probability, the fixed-offset curve, and its tail laws.  We now compare
the finite and limiting probabilities across the critical window to
resolve finite-size corrections, and separately simulate the
success/failure marginal of the ideal PGM to display sampling
fluctuations.  An independent frame calculation supplies the sampling
probabilities and checks the finite spectral formula.  These
computations illustrate the approach to the proven limit and
distinguish it from the sampling noise at a fixed trial count.
The calculations use
\begin{equation*}
 \begin{gathered}
 n\in\{4,8,12,16,24,32,64,128,256\},\\
 c=-3,-2,\ldots,7,
 \end{gathered}
\end{equation*}
with $k=n+c$.  These 99 parameter pairs include every residue class of
$k$ modulo four and hence test both the ordinary and exceptional branches.
All computations concern the uniform ensemble of pure stabilizer states
and recovery of the complete label under an ideal collective measurement,
as in Eq.~\eqref{eq:operational}.

\subsection{Finite probabilities and an independent frame calculation}
The finite theoretical curves are evaluated from
Eq.~\eqref{eqS:finite-success}.  Ensemble sizes, sector multiplicities,
and eigenvalues are first computed with exact integer and rational
arithmetic.  Decimal arithmetic enters when taking square roots and
summing the spectral amplitudes.  The dashed curves use
Eq.~\eqref{eq:curve}, with the products and series truncated as specified
below.

To obtain the probabilities used for sampling, we also compute the PGM
success directly from a symmetry restriction of the unnormalized frame
$S_{n,k}=\sum_x\rho_x^{\otimes k}$ in
Eq.~\eqref{eqS:gram-frame}.  This calculation uses affine supports and
binary-code counts, without importing the closed eigenvalues or
multiplicities of Theorem~\ref{thmS:finite-spectrum}.
Diagonal Clifford phases select binary row codes whose codewords have
weight divisible by four.  Linear changes of the $n$ computational
coordinates and stochastic orthogonal changes of the $k$ copy
coordinates then group the orbit sums by the code dimension $r\leq n$
and by whether the code contains $\one$.
Normalized orbit sums give a real symmetric matrix $B_{n,k}$ on an
invariant subspace containing the seed $|0^n\rangle^{\otimes k}$.

The matrix entries are assembled by counting affine supports and pairs
of row codes.  For two computational-basis $k$-tuples, let $a$ be the
affine dimension of their combined points in $\F_2^n$.
After the phase selection, their frame entry has weight
\begin{equation}
 W_a(n,k)=\sum_{m=a}^{n}
 \gbinom{n-a}{m-a}{2}\,2^{m(m+3)/2-mk}.
 \label{eq:num-support-weight}
\end{equation}
Here the Gaussian binomial counts the $m$-dimensional affine supports
containing that affine span; the remaining factor counts the stabilizer
phase choices and includes the $k$-copy amplitude normalization.
The implementation combines Eq.~\eqref{eq:num-support-weight} with exact
counts of the code pairs and normalizes the orbit sums before
diagonalization.  Across the grid, the reduced matrices have dimensions
between 1 and 258.

Let $B_{n,k}v_j=\eta_jv_j$, with orthonormal eigenvectors, and let
$v_j(0)$ denote the component along the normalized seed orbit.
The frame computation returns
\begin{equation}
 q_{n,c}=\left[\sum_{\eta_j>0}
             \frac{|v_j(0)|^2}{\sqrt{\eta_j}}\right]^2,
 \qquad k=n+c.
 \label{eq:num-frame-success}
\end{equation}
This is the squared diagonal matrix element of $S_{n,k}^{-1/2}$;
the inverse is restricted to its support.  By
Lemma~\ref{lemS:PGM}, it equals the optimal conditional and average
success probability.  The use of the unnormalized frame is essential:
Eq.~\eqref{eq:num-frame-success} requires no additional factor of $N_n$.

The recorded diagonalizations use 100--2606 decimal digits, depending on
the matrix.  At working precision $d$, eigenvalues below
$\eta_{\max}10^{-d+15}$ are discarded, and the projection of the seed
onto the retained eigenspaces is checked.  The largest recorded seed
projection error and relative eigenvalue-equation residual are
$8.6\times10^{-101}$ and $4.6\times10^{-101}$, respectively.
The saved frame probabilities agree with the finite spectral formula
to within $5\times10^{-71}$ at all 99 points.  This comparison checks two
computational constructions; the analytical justification remains the
preceding theorems.

\smallskip
\noindent\emph{Code and data availability.}
The versioned computational supplement is supplied in
\path{anc/critical_curve_supplement/} with this submission.  It provides the independent frame construction, saved high-precision
probabilities and diagnostics, all 99,000 binary trial outcomes, and the
figure programs~\cite{ComputationalSupplement}.
From that supplement directory, \texttt{python recompute\_frame.py}
rebuilds the reduced frame matrices from exact combinatorial counts and
diagonalizes them, without using stored probabilities or the closed
spectrum.  In contrast, \texttt{python reproduce.py --resample}, with a
new directory specified by \texttt{--output-dir}, regenerates the fixed-seed Bernoulli trials from the packaged frame
probabilities and redraws the figures; it does not recompute a frame
matrix or automatically read newly recomputed results.

The file \path{records/frame_diagnostics.csv} and the saved frame records
give per-point working precisions and residuals; the recomputation script
specifies the precision policy.
For the logarithmic comparison, the \texttt{--recompute-log} option to
\texttt{reproduce.py} rebuilds the differences using saved 90-digit probabilities for $n<256$,
reevaluates the $n=256$ formulas at 160 digits, and repeats three checks
at 200 digits.  The accompanying \path{precision_checks.json} records
these checks.
The English \path{REPRODUCIBILITY_GUIDE.pdf} gives installation instructions,
output locations, and the scope of each command.
Release-wide SHA-256 checksums and file sizes are recorded in
\path{release_manifest.json} and checked by
\texttt{python verify\_release.py}.  The supplement also records both the
original and release-validation environments, and preserves the historical
checksum manifests and the pre-release discrepancy investigation.

\subsection{Sampling procedure and confidence intervals}
At each parameter pair we generate $M=1000$ independent Bernoulli
outcomes with success probability $q_{n,c}$ from
Eq.~\eqref{eq:num-frame-success}, giving 99,000 binary outcomes in total.
Writing $s_{n,c}$ for the success count, the plotted frequency is
$\widehat p_{n,c}=s_{n,c}/M$.  Each point uses the seed tuple
$(20260920,n,k)$, passed to NumPy's \texttt{SeedSequence} for its
PCG64 generator.
All outcomes are retained, with no seed selection based on agreement
with theory.  The high-precision frame probability is converted to
floating point for Bernoulli sampling; the finite spectral formula is
used only as a comparator.
These trials simulate the success/failure marginal of the ideal PGM.
They do not construct a measurement circuit or sample its full output
label, and therefore provide no implementation-complexity claim.

The error bars are equal-tailed, pointwise 95\% Clopper--Pearson
intervals.  For a success count $s$, let $Q_{a,b}(z)$ be the
$z$-quantile of a beta distribution with parameters $a,b$.  The endpoints
are
\begin{align*}
 L(s)&=\begin{cases}
 0,&s=0,\\
 Q_{s,M-s+1}(0.025),&s>0,
 \end{cases}\\
 U(s)&=\begin{cases}
 1,&s=M,\\
 Q_{s+1,M-s}(0.975),&s<M.
 \end{cases}
\end{align*}
In the lower strips of Fig.~\ref{fig:num-sampling}, the signed residual
and its interval are
\begin{align}
 \Delta_{n,c}&=100(\widehat p_{n,c}-P^*_{n,n+c}),\notag\\
 I_{n,c}&=100\bigl([L(s_{n,c}),U(s_{n,c})]-P^*_{n,n+c}\bigr).
 \label{eq:num-residual}
\end{align}
Thus Eq.~\eqref{eq:num-residual} measures percentage points, rather
than relative percentage error.  The common vertical scale preserves
the sign and size of the fluctuations across panels.
The reported root mean square error is
\begin{equation*}
 R_n=100\left[\frac1{11}\sum_{c=-3}^{7}
              (\widehat p_{n,c}-q_{n,c})^2\right]^{1/2}
\end{equation*}
in percentage points.  It ranges from 0.52 to 0.91 across the nine
panels.  In this realization, 97 of the 99 intervals contain the
corresponding finite theoretical probability.  The intervals are
pointwise; they do not constitute a simultaneous 95\% confidence band.

\begin{figure*}[!tp]
 \centering
 \includegraphics[width=\textwidth,trim=0 67bp 0 0,clip]
 {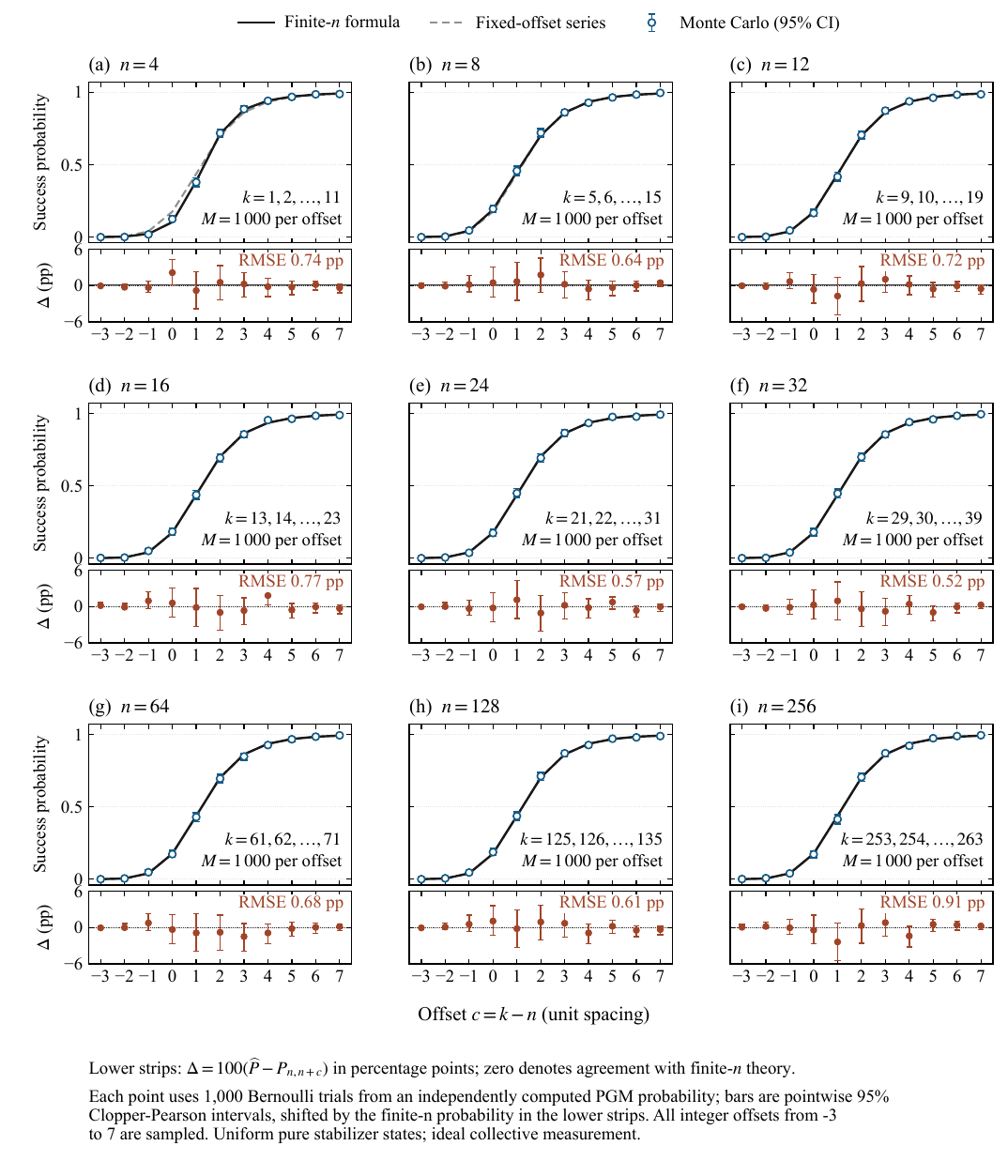}
 \caption{Sampling in the critical window $k=n+c$.  Each panel shows
 $P^*_{n,n+c}$ from the finite spectrum (solid), $P_\infty(c)$ from
 Eq.~\eqref{eq:curve} (dashed), and sample frequencies from 1000 Bernoulli
 trials per integer offset using an independently computed PGM frame
 probability (circles).
 Bars are pointwise 95\% Clopper--Pearson intervals.  Lower strips show
 the signed residual $\Delta_{n,c}$ and its translated interval in
 percentage points (pp); the displayed RMSE is $R_n$.
 Lines connect integer offsets and guide the eye.}
 \label{fig:num-sampling}
\end{figure*}

\subsection{Resolving finite-size corrections}
The sampling residual in Eq.~\eqref{eq:num-residual} is distinct from
the deterministic finite-size correction
\begin{equation}
 d_{n,c}:=P^*_{n,n+c}-P_\infty(c).
 \label{eq:num-finite-deviation}
\end{equation}
Figure~\ref{fig:num-log} compares two logarithmic displays.  A logarithmic
probability axis resolves the low-success side of the crossover, but
large-$n$ curves still overlap.  Plotting $|d_{n,c}|$ instead resolves
their separation from the limiting curve.  The absolute value suppresses
the direction of the correction; the accompanying data retain the
signed differences in Eq.~\eqref{eq:num-finite-deviation}.

For $n<256$, the finite and limiting probabilities are evaluated to
90 decimal digits, retaining $r=0,\ldots,100$ in the limiting series
and $j=1,\ldots,360$ in the products $A,Q$.
For $n=256$, the logarithmic comparison uses 160 digits,
$r=0,\ldots,130$, and products through $j=640$.
The products are accumulated with 12 additional guard digits.
The differences are formed in decimal arithmetic before conversion to
plotting precision.  This order matters: at $n=256$ and
$c=-2,2,6$, the deviations are of order $10^{-79}$--$10^{-78}$,
below the precision of the saved 70-digit frame probabilities.
These three differences were recomputed at 200 digits, with
$r=0,\ldots,150$ and products through $j=800$; the relative changes
were below $4\times10^{-81}$.  These are precision-stability checks,
rather than outward-rounded numerical enclosures.

\begin{figure*}[!tp]
 \centering
 \includegraphics[width=\textwidth,trim=0 33bp 0 0,clip]
 {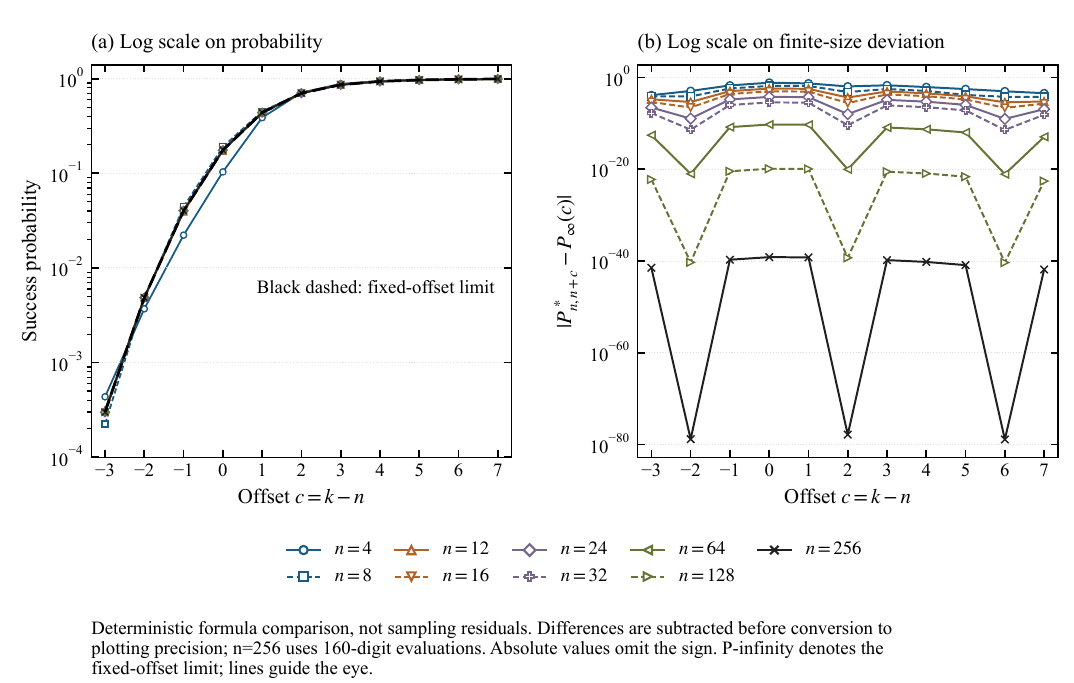}
 \caption{Deterministic comparison of the finite probabilities and the
 fixed-offset limit on the same grid as Fig.~\ref{fig:num-sampling}.
 (a) Success probabilities on a logarithmic axis; the black dashed line
 is $P_\infty(c)$.  (b) Absolute finite-size corrections
 $|P^*_{n,n+c}-P_\infty(c)|$ on a logarithmic axis, computed before
 conversion to plotting precision. The plotted differences come from high-precision
 formula evaluations.  Lines connect integer offsets and guide the eye.}
 \label{fig:num-log}
\end{figure*}

Table~\ref{tab:num-center} gives a numerical cross-section at $c=0$.
The finite probabilities need not approach the limit from one side:
$n=4,12$ lie below $P_\infty(0)$, whereas $n=8,16$ lie above it.
For $n\geq32$, the finite-size correction is already much smaller than
the sampling fluctuations at $M=1000$.  In particular, at the limiting
central probability the Bernoulli standard error is
\begin{equation*}
 \sqrt{\frac{P_\infty(0)(1-P_\infty(0))}{1000}}
 \simeq0.0120,
\end{equation*}
or about 1.20 percentage points, compared with
$|d_{32,0}|\simeq3.64\times10^{-6}$.
The Monte Carlo panels illustrate the observable sampling scale,
while the deterministic comparison resolves the much smaller
finite-size corrections.  The existence and value of the fixed-offset
limit follow from Theorem~\ref{thm:curve}.

\begin{table}[t]
 \centering
 \caption{Central probability and finite-size correction at $c=0$.
 The last column is the sampling RMSE over all eleven offsets, in
 percentage points.  Probabilities are rounded to ten decimal places.}
 \label{tab:num-center}
 \begin{tabular}{rccc}
 \toprule
 $n$ & $P^*_{n,n}$ & $|d_{n,0}|$ & $R_n$ (pp)\\
 \midrule
 4   & 0.1035661492 & $7.252\times10^{-2}$  & 0.74\\
 8   & 0.1902706150 & $1.419\times10^{-2}$  & 0.64\\
 12  & 0.1723097886 & $3.774\times10^{-3}$  & 0.72\\
 16  & 0.1770128693 & $9.291\times10^{-4}$  & 0.77\\
 24  & 0.1761420253 & $5.824\times10^{-5}$  & 0.57\\
 32  & 0.1760874300 & $3.640\times10^{-6}$  & 0.52\\
 64  & 0.1760837896 & $5.555\times10^{-11}$ & 0.68\\
 128 & 0.1760837896 & $1.293\times10^{-20}$ & 0.61\\
 256 & 0.1760837896 & $7.011\times10^{-40}$ & 0.91\\
 \bottomrule
 \end{tabular}
\end{table}


\section{Conclusion}
\label{sec:conclusion}

We have determined the exact finite-copy optimum for identifying a uniformly distributed pure $n$-qubit stabilizer state from $k$ identical copies under arbitrary collective measurements.  The pretty-good measurement is optimal for every finite pair $(n,k)$, including the linearly dependent regime, and its success probability is obtained from the complete positive spectrum of the stabilizer Gram matrix.  The spectral decomposition includes both the ordinary defect sectors and the exceptional all-ones sectors that occur when $4\mid k$; sectors of zero multiplicity at a given physical dimension are removed only after specialization.  Thus the result is not restricted to the stable commutant range and gives an exact formula at every physical dimension $D=2^n$.

The proof separates the spectral problem into several structural components.  Clifford tensor-power duality expresses the averaged frame operator in the algebra of stochastic orthogonal relations.  A defect filtration then identifies the physical graded sectors and separates the ordinary and exceptional branches.  Mixed stabilizer moments and weighted-flag inversion determine the ordinary eigenvalues, while shortening reduces the exceptional branch to the ordinary spectrum at one fewer copy.  Invariant-code counting gives the corresponding multiplicities.  Finally, polynomial continuation of the finite candidate spectrum extends the stable-range calculation to all physical dimensions without discarding labels prematurely.  This separation is important because the eigenvalue formulas, the multiplicities, and the passage from the stable range to arbitrary $n$ require logically distinct arguments.

As a consequence of the finite spectrum, the full integer-scale critical window can be resolved.  For $k=n+c$ with fixed $c\in\mathbb Z$, the optimal success probability converges to the explicit series $P_\infty(c)$.  In particular, the crossover is genuinely nontrivial at unit offsets: it is neither a step function nor a consequence of order-level copy bounds alone.  The limiting curve is normalized by its first spectral moment, and the exchange of the finite-sector limit with the infinite sum follows from uniform tail control.  The exceptional contribution vanishes in this limit, although it remains necessary for the exact finite-copy spectrum.

The two tails of the limiting curve exhibit different but explicit behavior.  Below threshold, the success probability has a quadratic-exponential decay governed by the dominant left-tail sector.  Above threshold, the failure probability satisfies
\[
  1-P_\infty(c)\sim 2^{-c},
  \qquad c\to+\infty,
\]
with leading coefficient one.  Hence the exact curve refines the statement that the sample complexity is $n+O(1)$: it specifies both the finite offset dependence and the asymptotic cost of increasing the target success probability within the critical window.

The present work concerns the uniform ensemble of pure qubit stabilizer states and unrestricted collective measurements.  It does not by itself provide a restricted-depth, few-copy, or otherwise implementation-constrained realization of the optimal measurement.  Such constraints define a different optimization problem and need not preserve the exact spectrum or the PGM optimality proved here.  Natural further questions include determining which parts of the defect-filtration method persist for nonuniform priors, noisy stabilizer ensembles, restricted measurement architectures, and odd-prime qudit stabilizer states.  The finite-dimensional continuation argument may also be useful more broadly when a commutant calculation is first available only in a stable range but the physically relevant dimensions extend beyond it.

\section*{Declaration of AI Assistance}
AI-assisted tools were used extensively during proof exploration, technical derivation, computational checking, and manuscript preparation. The human authors reviewed the mathematical statements and proofs included in the manuscript and take responsibility for the content of the paper.

\clearpage
\onecolumn
\twocolumn
\appendices
\renewcommand{\thesubsection}{\thesection.\Roman{subsection}}
\renewcommand{\thesubsectiondis}{\appendixname~\thesubsection.}

\section{Optimal measurement and minimax success}
\label{app:measurement}

\subsection{Proof of Lemma~\ref{lemS:PGM}}
\label{app:proof-pgm}

\begin{proof}
Put $S=\sum_x|\phi_x\rangle\!\langle\phi_x|$ and
$K=\ran S$.  All inverses in this proof are taken on $K$.  The PGM vectors
are $|\mu_x\rangle=S^{-1/2}|\phi_x\rangle$, and
\begin{equation}
 \sum_x|\mu_x\rangle\!\langle\mu_x|=\Pi_K.
 \label{eqS:PGMresolution}
\end{equation}
The covariance of $S^{-1/2}$ implies that
\begin{equation}
 a_x:=\langle\phi_x|S^{-1/2}|\phi_x\rangle=a
 \label{eqS:a-constant}
\end{equation}
is independent of $x$.  Taking a trace gives
$Na=\Tr(S^{-1/2}S)=\Tr\sqrt S$.  The dual operator associated with the PGM
is
\begin{equation}
 \Gamma:=\frac1N\sum_x
 |\phi_x\rangle\!\langle\phi_x|
 |\mu_x\rangle\!\langle\mu_x|
 =\frac{a}{N}S^{1/2}.
 \label{eqS:dual}
\end{equation}
For a positive definite $A$ on $K$, the rank-one criterion
$A\geq|v\rangle\!\langle v|$ is equivalent to
$\langle v|A^{-1}|v\rangle\leq1$.  Applying it with
$A=aS^{1/2}$ gives equality:
\begin{equation}
 \langle\phi_x|(aS^{1/2})^{-1}|\phi_x\rangle
 =a^{-1}\langle\phi_x|S^{-1/2}|\phi_x\rangle=1.
 \label{eqS:rankone}
\end{equation}
Consequently $\Gamma-N^{-1}|\phi_x\rangle\!\langle\phi_x|\geq0$ for every
$x$.  Together with Eq.~\eqref{eqS:dual}, these are the
Yuen--Kennedy--Lax optimality conditions
\cite{YuenKennedyLax1975,EldarMegretskiVerghese2004}.  Thus the PGM is
optimal.  Finally, polar decomposition gives
$\langle\mu_x|\phi_y\rangle=(\sqrt H)_{xy}$, and transitivity makes the
diagonal entries of $\sqrt H$ equal.  Therefore
$P_{\rm PGM}=(\Tr\sqrt H/N)^2$.
\end{proof}

\smallskip\noindent\emph{Minimax consequence.}\phantomsection
\label{app:minimax}

The same measurement also optimizes the worst-case conditional success
probability for a transitive pure-state orbit.  Indeed, every PGM outcome
has conditional success $a^2$, so its worst-case and average success agree.
For any other POVM, the worst-case success is at most its uniform average,
which is at most the optimal average in Lemma~\ref{lemS:PGM}.  Thus the
uniform-prior value used here is also the unrestricted minimax value.

\section{Physical code sectors}
\label{app:filtration}

\subsection{Proof of Lemma~\ref{lemS:orbit-tail}}
\label{app:proof-orbit-tail}

\begin{proof}
The code isometry of Lemma III.8 of
Ref.~\cite{MontealegreGross2025} identifies $C_N$ with the $s$-fold
mixed-signature model for $(Q,q_N)$; Lemma III.11 identifies the quotient
action with its stochastic group $G_N$.  Since $s\leq k\leq n$, Theorem
V.6 and Lemma V.18 of that reference apply inside this quotient and give
the actual orthogonal Clifford-module decomposition
\begin{equation}
 C_N^{G_N}\simeq
 \bigoplus_{[L]}\eta_{Q,L}(\bm1),
 \label{eqS:quotient-invariant-decomposition}
\end{equation}
where $[L]$ runs once over the ordinary and exceptional stochastic
isotropic-defect orbits of $Q$.  This is an actual orthogonal decomposition
of represented subspaces of $C_N$, not a relation-algebra dimension count.

We identify the entire tail in Eq.~\eqref{eqS:quotient-invariant-decomposition}.
Whenever it is nonempty, choose an ordinary isotropic line
$\ell<Q$ and form $Q'=\ell^\perp/\ell$, of dimension $s-2$.
Inverse image under $\ell^\perp\to Q'$ sends a defect $L'<Q'$ to a
defect $\widetilde L<Q$ containing $\ell$.  Conversely, every nonzero
ordinary defect orbit of $Q$ has a representative containing $\ell$, and
every exceptional orbit of dimension at least two has one as well:
apply the single-orbit/Witt-extension Lemma III.4 of the same reference
to an ordinary line in that defect.  That lemma says precisely that orbit
classes are determined by dimension and by containment of the Wu vector.
Thus quotient and inverse image are mutual on orbit classes, preserve Wu
containment, and give the
orbit bijection
\begin{equation}
 \begin{aligned}
 \{\text{defect orbits of }&Q'\} \longleftrightarrow\{\text{defect orbits of }Q\}
   \setminus\mathcal E_Q,\\
 &\mathcal E_Q:=\begin{cases}
 \{[0],[\langle[\one_k]_N\rangle]\},&4\mid k,\\
 \{[0]\},&4\nmid k.
 \end{cases}
 \end{aligned}
 \label{eqS:defect-orbit-tail-bijection}
\end{equation}
The excluded Wu line is isotropic exactly when
$q_N([\one_k]_N)=q_k(\one_k)=k=0\pmod4$.

For each corresponding pair, let
$\pi_N:N^\perp\to Q$, set
$N_\ell:=\pi_N^{-1}(\ell)$, and set
$N_L:=\pi_N^{-1}(\widetilde L)$.  Then
\begin{equation}
 C_{N_L}\subset C_{N_\ell}\subset C_N,\qquad
 N_L^\perp/N_L
 \simeq(\widetilde L^\perp/\widetilde L)
 \simeq(L'^\perp/L').
 \label{eqS:double-quotient-code}
\end{equation}
We use the Clifford rank convention of Sec.~IV.1 of
Ref.~\cite{MontealegreGross2025}: rank is the largest rank of a symmetric
weight form for the diagonal Clifford subgroup.  In a quotient code, a
computational label $f$ has weight form $f^*b_Q$, the pullback of the
quotient polar form.  Write $\Delta^{(j)}$ for the rank-$j$ Clifford
component.
Apply Lemma III.8 first to the ordinary defect $N$ and again to the
ordinary defect $N_\ell$.  Under these two code isometries, the single
physical subcode $C_{N_L}$ is, respectively, the quotient code
$C_{\widetilde L}^{(Q)}$ and the quotient code $C_{L'}^{(Q')}$.
This can also be read directly from the CSS coset
basis~\eqref{eqS:CSS-coset}: the direct and iterated embeddings are both
the normalized uniform sum over
$\operatorname{Hom}(\F_2^n,N_L)$.  Lemma V.10 now identifies, in both
descriptions, the same intrinsic maximal-rank subspace
\begin{equation}
 C_{N_L}\cap
 \Delta^{(\dim(N_L^\perp/N_L))}.
 \label{eqS:same-maximal-rank-subspace}
\end{equation}
The trivial residual stochastic isotype in Theorem V.1, and its induced
placement in Theorem V.14, therefore identify the paired modules
$\eta_{Q,\widetilde L}(\bm1)\simeq\eta_{Q',L'}(\bm1)$.
Summing the right side of
Eq.~\eqref{eqS:defect-orbit-tail-bijection} and applying
Eq.~\eqref{eqS:quotient-invariant-decomposition} to $Q'$ gives exactly
$a_{s-2}(D)$.  If no ordinary line exists, both tails are empty and this
is the stated terminal convention.  Separating the one or two head
orbits proves Eq.~\eqref{eqS:orbit-tail-recursion}.
\end{proof}

\subsection{Proof of Proposition~\ref{propS:filtration-bridge}}
\label{app:proof-filtration-bridge}

\begin{proof}
The cited decomposition is applied directly to represented subspaces of the
physical tensor-power Hilbert space.  The complete logical chain is as
follows.

\begin{enumerate}
\item \emph{Local filtration and direct sums.}
Here is the precise dictionary with Ref.~\cite{MontealegreGross2025}.  Its
quadratic space $T$ is our $(E_k,q_k,\one_k)$, its
$\operatorname{St}(T)$ is $G$, and its $C_N$, $\mathcal K_r$, and
$\mathcal L_r$ are, respectively, our $C_N$, $\mathcal K_{k,r}$, and
$\mathcal L_{k,r}$.  Its $C'_r$ and $D_r$ are sums of code projectors whose
ranges are our $\mathcal F'_r$ and $\mathcal F_r^{\Exc}$.  With these
projector-sum conventions, the individual-projector identity above gives
\begin{equation}
 P_{\one}C'_r=2^rD_{r+1}.
 \label{eqS:projector-sum-normalization}
\end{equation}
Indeed, each exceptional $(r+1)$-space has $2^r$ ordinary hyperplanes
missing $\one_k$, the graphs of linear maps from a fixed complement to
$\langle\one_k\rangle$.  This supplies the scalar missing from identity
(31) of that reference.  Taking ranges gives the required unchanged relation
$P_{\one}\mathcal F'_r=\mathcal F_{r+1}^{\Exc}$.
Lemma V.7 supplies ideal/kernel invariance on the associated grades;
Lemma V.9 gives local spanning, using the represented defect ideals and
compressed action of Proposition III.5 and Lemma III.6; Lemma V.10 identifies
each local grade with the maximal-rank part of the actual code $C_N$; and
Lemma V.12 proves that these local subspaces are linearly independent.
Spanning plus that independence gives exactly the global represented direct sums
in Eq.~\eqref{eqS:local-grades}; pairwise inequivalence is supplied by the
induced-module theorem used in the next step, not inferred from dimension
counting alone.

\item \emph{Induced stochastic decomposition.}
The quotient action $P_N\to G_N$ is the surjection of Lemma III.11 of that
reference.  With this action, Theorem V.14 proves the individual induced-grade
formula~\eqref{eqS:induced-grade} when its residual-rank hypothesis
$s=k-2r<n$ holds, including pairwise inequivalence of the displayed Clifford
modules.  Theorem V.6 assembles those grades into the full represented
decomposition under the hypothesis $k\leq n$, which is precisely our stable
range; its stated range also includes the remaining case
$s=n$ (necessarily $n=k$ and $r=0$).  Both results use the present stochastic
quadratic space $(E_k,q_k,\one_k)$, and Theorem V.6 likewise includes the
exceptional filtration when $4\mid k$.

\item \emph{Fixed parts are represented modules.}
Frobenius reciprocity gives
\begin{equation}
 \begin{aligned}
 \dim\!\left[\operatorname{Ind}_{P_N}^{G}(\tau)\right]^G
 &=\dim\operatorname{Hom}_{P_N}(\bm1,\tau)\\
 &=\begin{cases}1,&\tau=\bm1,\\0,&\tau\ne\bm1.\end{cases}
 \end{aligned}
 \label{eqS:Frobenius}
\end{equation}
and proves Eq.~\eqref{eqS:fixed-grade}.  Because
Eq.~\eqref{eqS:induced-grade} is a decomposition of the represented code
subspace itself, each surviving $G$-fixed factor is an actual
Clifford-invariant subspace there, not an abstract multiplicity space.

\item \emph{Adjacent differences and the exceptional head.}
Under the quotient-code isometry, Lemma V.10 identifies the zero-defect
module $\eta_{Q,0}(\bm1)$ in Lemma~\ref{lemS:orbit-tail} with
$\eta^{\Ord}_{k,r}(\bm1)$.  When $4\mid k$, inverse image of the quotient Wu
line is $N+\langle\one_k\rangle$.  Lemma V.10 and Theorem V.14 applied to
that local exceptional grade, together with the range identity just proved
and the direct sum of Lemma V.12, identify it with
$\eta^{\Exc}_{k,r+1}(\bm1)$.  Combining these identifications with
Eq.~\eqref{eqS:orbit-tail-recursion} gives
Eqs.~\eqref{eqS:adjacent-difference} and
\eqref{eqS:ordinary-exceptional-split}.  The orbit-tail bijection is complete
at the terminal conventions stated in Lemma~\ref{lemS:orbit-tail}; hence the
shift by one neither omits nor duplicates an all-ones boundary sector.

\item \emph{Multiplicity-free exhaustion.}
Lemma V.18 of Ref.~\cite{MontealegreGross2025} identifies
$\operatorname{span}\{|S\rangle^{\otimes k}:S\in\Stab_n\}$ with the
trivial-$G$ isotype and decomposes it as one copy of $\eta(\bm1)$ for every
nonempty ordinary or exceptional defect orbit.  Together with the preceding
terminal audit, this proves multiplicity-free exhaustion and absence of
overlap.  No faithfulness of a formal relation algebra is used: the cited
theorems concern the represented physical code subspaces, and any abstract
character killed by physical compression contributes no nonempty module to
this sum.

\item \emph{Vanishing on the complement and scalar action.}
Set
\begin{equation}
 \Pi_k:=|O_k^{\rm st}(2)|^{-1}
 \sum_{O\in O_k^{\rm st}(2)}R_n(O).
\end{equation}
This is the projector onto the trivial $G$-isotype.  Left and right
multiplication by every defect-zero relation permute the complete set
$\Sigma_{k,k}(2)$, so the represented operator identity is
\begin{equation}
 A_k(D)=\Pi_kA_k(D)\Pi_k.
 \label{eqS:bi-invariance}
\end{equation}
It annihilates every nontrivial $G$-type.  In addition,
$A_k(D)=Z_{n,k}\rho_{n,k}$, so its range lies in the stabilizer-tensor span
and it annihilates that span's orthogonal complement.  Finally, $A_k(D)$
commutes with the Clifford action; Schur's lemma therefore makes it scalar on
each single module in Eq.~\eqref{eqS:fixed-grade}.
\end{enumerate}

Only the represented-module decomposition in this proof uses $n\geq k$.
For $n<k$, Lemma~\ref{lemS:polynomial-continuation}, proved in
Appendix~\ref{app:proof-continuation}, continues the spectral multiset
through actual trace moments; it does not assert an all-$n$
physical-module decomposition.
\end{proof}

\section{Ordinary eigenvalues}
\label{app:ordinary}

\subsection{Proof of Lemma~\ref{lemS:mixed-moment}}
\label{app:proof-mixed-moment}

\begin{proof}
Start with the ordinary $t$-copy identity
\eqref{eqS:moment-GNW}.  Write a row and column index as
$(x_A,x_B;y_A,y_B)$ and define the coordinate switch
\begin{equation}
 \tau_B(x_A,x_B;y_A,y_B):=(x_A,y_B;y_A,x_B).
 \label{eqS:coordinate-switch}
\end{equation}
It is an involution, fixes the diagonal pointwise, and obeys
\begin{equation}
 \begin{aligned}
 &(q_t\oplus-q_t)(x_A,x_B;y_A,y_B)\\
 &\quad=(q_{a,b}\oplus-q_{a,b})\tau_B(x_A,x_B;y_A,y_B).
 \end{aligned}
\end{equation}
It therefore bijects $\Sigma_{t,t}(2)$ with
$\Sigma(q_{a,b},q_{a,b})$ and preserves
$\dim(T\cap\Delta_V)$.  Partial transpose in the $B$ tensor slots, in the
computational basis, gives entry by entry,
\begin{equation}
 \Theta_B(r(T))=r_{q_{a,b}}(\tau_BT),
 \qquad
 \Theta_B(R_n(T))=R_{n,q_{a,b}}(\tau_BT).
 \label{eqS:partial-transpose-relations}
\end{equation}
On the other side it sends
$\rho_S^{\otimes t}$ to
$\rho_S^{\otimes a}\otimes\overline{\rho_S}^{\otimes b}$, the projector
onto $|S;q_{a,b}\rangle$.  Applying $\Theta_B$ to
Eq.~\eqref{eqS:moment-GNW} proves
Eq.~\eqref{eqS:mixed-moment-lemma}, with the same coefficient for every
relation and with no stable-range restriction.  Changing $U_S$ only
changes the represented stabilizer ray by a phase, which disappears from
this projector.

For completeness, in the stable range $t\leq n-1$ the transported
relations are the linearly independent commutant basis of Proposition
III.3 of Ref.~\cite{MontealegreGross2025}; thus the coefficient separation
is unique there.  Outside that range Eq.~\eqref{eqS:partial-transpose-relations}
transports the all-$n$ identity of Theorem 5.3 of
Ref.~\cite{GrossNezamiWalter2021}, so no independence assumption is being
made.  The same involution transports the ordinary relative-diagonal
count, Eq. (4.8) in the proof of Theorem 4.10 there, giving
Eq.~\eqref{eqS:mixed-diagonal-count}.  Hence
\begin{align}
 \sum_T\Tr R_{n,q_{a,b}}(T)
 &=\sum_{\ell=0}^{t-1}2^{\binom\ell2}
   \gbinom{t-1}{\ell}{2}D^{t-\ell}\notag\\
 &=D\prod_{j=0}^{t-2}(D+2^j)=Z_{n,t},
\end{align}
including $t=1$ by the empty-product convention.  This independently
checks the normalization.  Conjugating every relation operator by the
basis permutation induced by a quadratic-space isometry proves the last
assertion.
\end{proof}

\subsection{Proof of Lemma~\ref{lemS:weighted-flags}}
\label{app:proof-weighted-flags}

\begin{proof}
The supports are
\begin{equation}
 F+\operatorname{Hom}(\F_2^n,N),\qquad
 OF+\operatorname{Hom}(\F_2^n,M).
\end{equation}
Equality of one vector in the two supports implies
$(O-I)F\in\operatorname{Hom}(\F_2^n,N+M)$, and the converse follows by
lifting a map into $N+M$ through the surjection $N\oplus M\to N+M$.
Since $\operatorname{im}F=N^\perp$, this is exactly
Eq.~\eqref{eqS:weighted-overlap-condition}.  A nonempty intersection is a
coset of $\operatorname{Hom}(\F_2^n,K)$ and has $D^{\dim K}$ points.
The two CSS normalizations contribute $D^{-r}$, proving the overlap
$D^{\dim K-r}=D^{-\ell}$ in both directions.

The left and right kernels of the pairing on $N/K\times M/K$ are
$(N\cap M^\perp)/K$ and $(M\cap N^\perp)/K$.  Thus it is nondegenerate
exactly when both intersections equal $K$.  We now derive this condition
from the affine overlap, rather than assume it.  If
Eq.~\eqref{eqS:weighted-overlap-condition} holds, then for every
$x\in N^\perp$ one can write
\begin{equation}
 (O-I)x=n_x+m_x,\qquad n_x\in N,\quad m_x\in M.
\end{equation}
As $Ox\in M^\perp$ and $M\subseteq M^\perp$, this gives
$x+n_x=Ox+m_x\in M^\perp$.  The restrictions to $M/K$ of the functionals
$b(x,\,\cdot\,)$ with $x\in N^\perp$ span all of $(M/K)^*$: the
annihilator of $N^\perp$ inside $M$ is
$M\cap(N^\perp)^\perp=M\cap N=K$, so this restriction map has rank
$\dim(M/K)=\ell$.  Since every one of these functionals is also the
restriction of some $n_x\in N$, the map
\begin{equation}
N/K\longrightarrow(M/K)^*,\qquad[n]\longmapsto b(n,\,\cdot\,)
\end{equation}
is surjective.  Both spaces have dimension $\ell$, hence it is an
isomorphism.  Equivalently, the unquotiented map $N\to(M/K)^*$ has kernel
$N\cap M^\perp=K$; because the pairing is square, its right kernel is also
zero and $M\cap N^\perp=K$.  Thus the pairing is nondegenerate.

Conversely, whenever the pairing is nondegenerate the displayed map is an
isomorphism.  For $x\in N^\perp$, the functional
$b(x,\,\cdot\,)$ on $M$ vanishes on $K$, so it selects a unique class
$[n_x]\in N/K$ satisfying $b(n_x,m)=b(x,m)$ for every $m\in M$.
Equivalently $x+n_x\in M^\perp$.  This proves existence and uniqueness in
Eq.~\eqref{eqS:weighted-canonical-map}.  Replacing $x$ by $x+n$ with
$n\in N$ changes $n_x$ by the same class, and changing $n_x$ by $K$
changes the output by $M$; hence $J_{M,N}$ is well defined.  It is linear
and bijective by symmetry.  Because $x\perp N$ and $N$ is totally
isotropic,
\begin{equation}
 q_k(x+n_x)=q_k(x)+q_k(n_x)+2b(x,n_x)=q_k(x),
\end{equation}
so it preserves $q_k$ and therefore its polar form.  It fixes the
stochastic vector: $b(\one_k,M)=0$ permits $n_{\one_k}=0$.

Under Eq.~\eqref{eqS:weighted-overlap-condition}, the decomposition above
has $[Ox]_M=[x+n_x]_M$, so $O$ induces $J_{M,N}$.  Conversely, if $O$
induces $J_{M,N}$, then
$[Ox]_M=[x+n_x]_M$ for all $x\in N^\perp$, which is
Eq.~\eqref{eqS:weighted-overlap-condition}.  This proves both directions.
Lemma III.4 of
Ref.~\cite{MontealegreGross2025} supplies one stochastic isometry carrying
$N$ to $M$, and Lemma III.11 makes the stabilizer map onto $G_N$.
Composing the two proves that every stochastic quotient isometry has one
coset of the kernel as its lift fibre, of cardinality $\kappa_N$.  The
direct quotient-lift statement is Lemma 4.20 of
Ref.~\cite{GrossNezamiWalter2021}.  In its qubit Witt-extension step
(Lemma 4.18 there), one must also check the mod-four obstruction.  Here the
extension is already a $q_k$-isometry fixing $\one_k$, and for
$I=\{x:b(x,x)=0\}=\one_k^\perp$ one has
$I^\perp=\langle\one_k\rangle$; hence both extended defect spaces meet
$I^\perp$ in the same fixed Wu line.  All hypotheses are therefore
satisfied.

It remains to prove Eq.~\eqref{eqS:weighted-flag-count}, including
uniqueness.  Pass to $\overline E=K^\perp/K$, put $U=N/K$, and choose a
stochastic Witt decomposition
\begin{equation}
 \overline E=U\oplus W\oplus U^\vee,
 \qquad \dim W=s,\qquad\one_k\in W.
\end{equation}
The projection $M/K\to U^\vee$ is an isomorphism exactly when the displayed
pairing is nondegenerate.  For dual bases $u_i,v_i$ this forces, rather than
merely permits, the unique normalized basis
\begin{equation}
 m_i=v_i+x_i+\sum_{j=1}^{\ell}c_{ji}u_j,
 \qquad x_i\in W_0:=W\cap\one_k^\perp.
 \label{eqS:weighted-normal-form}
\end{equation}
The stochastic condition gives $x_i\in W_0$, while total $q_k$-isotropy is
equivalent to
\begin{equation}
 c_{ii}=q_k(x_i)/2,
 \qquad c_{ij}+c_{ji}=b(x_i,x_j)\quad(i<j).
 \label{eqS:weighted-normal-form-equations}
\end{equation}
Here $q_k(x_i)$ is even because $x_i\perp\one_k$.  Thus the $x_i$ give
$2^{\ell(s-1)}$ choices, the strict upper triangle of $C$ gives
$2^{\binom\ell2}$ choices, and all other entries are forced.  Conversely
substitution verifies the stochastic and quadratic equations, so no
objects are missed.  When $\ell=0$ the unique choice is $M=N$; when
$W_0=0$ all $x_i$ vanish and the same equations retain exactly the
alternating-matrix choices.  These also cover the terminal Witt index.
\end{proof}

\subsection{Proof of Lemma~\ref{lemS:flags}}
\label{app:proof-flags}

\begin{proof}
We compute the normalization lost by using the nonorthogonal family of CSS
codes.  Regard a $k\times n$ binary matrix as a map
$F:\F_2^n\to E_k$.  Choose $F$ with image $N^\perp$ and set
\begin{equation}
 |\psi_{N,F}\rangle
 :=D^{-r/2}\sum_{Z\in\operatorname{Hom}(\F_2^n,N)}|F+Z\rangle.
 \label{eqS:max-rank-coset}
\end{equation}
This is a normalized CSS coset vector.  Its quotient map
$f=[F]:\F_2^n\to Q_N=N^\perp/N$ is onto.  Since the polar form $b_Q$
on $Q_N$ is nondegenerate,
\begin{equation}
 \rank(f^*b_Q)=s=k-2r.
 \label{eqS:surjective-weight-rank}
\end{equation}
In matrix form this is $\rank(f^{\mathsf T}B_Qf)=s$: $f$ is surjective,
$B_Q$ is invertible, and $f^{\mathsf T}$ is injective.  Thus this weight
vector lies in the maximal-rank component of $C_N$.  Lemma V.10 of
Ref.~\cite{MontealegreGross2025}, whose hypothesis $s\leq n$ holds here,
identifies that component with the local ordinary grade $\mathcal K_N$.
As a support check, if $L\supsetneq N$ is a
strict isotropic extension, every computational string in $C_L$ has
quotient image contained in the proper subspace $L^\perp/N<Q_N$, whereas
every string in the support of $|\psi_{N,F}\rangle$ induces the surjection
$[F]:\F_2^n\to Q_N$.  The supports are therefore disjoint, so this vector
is orthogonal to every code $C_L$ from a strict extension $L\supsetneq N$.

Write $G=O_k^{\rm st}(2)$, $P_N=\operatorname{Stab}_G(N)$, and
\begin{equation}
 \kappa_N:=|\ker(P_N\longrightarrow G_N)|.
 \label{eqS:kappa}
\end{equation}
We first evaluate the spherical Gram row
\begin{equation}
 S_N(D):=\sum_{O\in G}
 \langle\psi_{N,F}|R_n(O)|\psi_{N,F}\rangle.
 \label{eqS:spherical-row}
\end{equation}
Lemma~\ref{lemS:weighted-flags} proves, in both directions, the support
criterion, its overlap weight, the canonical quotient isometry, and the
complete lift fibre.  There are $\gbinom{r}{\ell}{2}$ choices of
$K<N$ of codimension $\ell$ and, for each $K$,
$2^{\ell(s-1)+\binom\ell2}$ such $M$.  Combining these counts gives
\begin{align}
 \frac{S_N(D)}{\kappa_N}
 &=\sum_{\ell=0}^{r}\gbinom{r}{\ell}{2}
   2^{\ell(s-1)+\binom\ell2}D^{-\ell}\notag\\
 &=D^{-r}\prod_{j=0}^{r-1}(D+2^{s-1+j})
 =:\gamma_{k,r}(D),
 \label{eqS:H-qbinomial}
\end{align}
by the finite binary $q$-binomial theorem.  This overlap-weighted sum supplies
the normalization needed because the CSS code subspaces are not mutually
orthogonal.

We next compute the diagonal matrix element of $A_k(D)$.  Let
$P_N^{\rm CSS}$ denote the orthogonal projector onto $C_N$, and let $q_N$
be the nondegenerate $\mathbb Z_4$-quadratic form induced on
$Q_N=N^\perp/N$.  The code restriction of
Ref.~\cite{MontealegreGross2025} intertwines the Clifford action on $C_N$
with the mixed-signature tensor-power model $\Delta_{q_N}$; in general this
model is \emph{not} the pure tensor power on $s=k-2r$ copies.  Write
\begin{equation}
 \widetilde A_N(D):=
 \sum_{T\in\Sigma(q_N,q_N)}R_{n,q_N}(T).
 \label{eqS:quotient-relation-sum}
\end{equation}
Here $\Sigma(q_N,q_N)$ is the set of $s$-dimensional subspaces of
$Q_N\oplus Q_N$ on which $q_N(x)-q_N(y)$ vanishes and which contain
$([\one_k]_N,[\one_k]_N)$; write
$\Delta_{Q_N}:=\{(x,x):x\in Q_N\}$.
For a projective stabilizer state $|S\rangle=U_S|0^n\rangle$, put
$|S;q_N\rangle:=\Delta_{q_N}(U_S)|0^n\rangle^{\otimes s}$; its rank-one
projector is independent of the choice of $U_S$.  The code isometry
$\iota_N$ obeys
\begin{equation}
 \iota_NP_N^{\rm CSS}|S\rangle^{\otimes k}
 \!\langle S|^{\otimes k}P_N^{\rm CSS}\iota_N^{-1}
 =D^{-r}|S;q_N\rangle\!\langle S;q_N|.
 \label{eqS:CSS-stabilizer-compression}
\end{equation}
For $S=|0^n\rangle$ (choose $U_S=I$) this is the normalization of the
zero coset, and the general case follows by Clifford covariance.  This is
Lemma III.8 of
Ref.~\cite{MontealegreGross2025}.  Proposition II.13 of the same reference,
together with its equivalent-signature convention following Lemma III.1,
identifies $q_N$ with a signed Hamming form of total dimension $s$.
Lemma~\ref{lemS:mixed-moment} therefore applies and gives, for the full
parameter range,
\begin{equation}
 \frac1{N_n}\sum_{S\in\Stab_n}|S;q_N\rangle\!\langle S;q_N|
 =\frac{\widetilde A_N(D)}{Z_{n,s}}.
 \label{eqS:mixed-moment}
\end{equation}
Combining Eqs.~\eqref{eqS:moment-GNW},
\eqref{eqS:CSS-stabilizer-compression}, and~\eqref{eqS:mixed-moment} gives
the exact compression
\begin{equation}
 \begin{gathered}
 \iota_NP_N^{\rm CSS}A_k(D)P_N^{\rm CSS}\iota_N^{-1}
 =\alpha_{k,r}(D)\widetilde A_N(D),\\
 \alpha_{k,r}(D):=D^{-r}\frac{Z_{n,k}}{Z_{n,s}}.
 \end{gathered}
 \label{eqS:CSS-A-compression}
\end{equation}
The vector $\iota_N|\psi_{N,F}\rangle=|f\rangle$ is labelled by a
surjective map $f:\F_2^n\to Q_N$.  Thus
$\langle f|\widetilde A_N(D)|f\rangle=1$: a relation contributing to this
diagonal matrix element contains the full diagonal
$\{(f(v),f(v)):v\in\F_2^n\}=\Delta_{Q_N}$, and the unique
$s$-dimensional Lagrangian relation containing it is $\Delta_{Q_N}$ itself.
Consequently
\begin{equation}
 \langle\psi_{N,F}|A_k(D)|\psi_{N,F}\rangle
 =\alpha_{k,r}(D).
 \label{eqS:A-diagonal-alpha}
\end{equation}

Let $\theta_{k,r}(D)$ be the scalar of $A_k(D)$ on the fixed sector in
question.  Since $A_k=\Pi_kA_k\Pi_k$ and that fixed sector contains a single
Clifford module, Eqs.~\eqref{eqS:spherical-row} and
\eqref{eqS:A-diagonal-alpha} give the \emph{master scalar identity}
\begin{equation}
 \alpha_{k,r}(D)
 =\theta_{k,r}(D)
   \langle\psi_{N,F}|\Pi_k|\psi_{N,F}\rangle
 =\theta_{k,r}(D)\frac{\kappa_N\gamma_{k,r}(D)}{|G|}.
 \label{eqS:spherical-ratio}
\end{equation}
Every factor is now accounted for at the represented-operator level:
$\alpha_{k,r}$ comes from the mixed-signature moment and exact CSS
compression~\eqref{eqS:CSS-A-compression}, with its diagonal entry fixed by
Eq.~\eqref{eqS:A-diagonal-alpha}; $\gamma_{k,r}$ is the complete affine CSS
overlap and weighted-defect-flag sum~\eqref{eqS:H-qbinomial}; $\kappa_N$ is
the full quotient-isometry lift fibre from
Lemma~\ref{lemS:weighted-flags}; and $|G|^{-1}$ is the normalization of the
physical projector $\Pi_k$.  In particular, the flag variable is
$\ell=\operatorname{codim}_N(N\cap ON)$ and necessarily runs through the
complete range $0\leq\ell\leq r$ in Eq.~\eqref{eqS:H-qbinomial}; its upper
limit is not an ansatz.
Finally,
$|G|/\kappa_N=|\cG_{k,r}|\,|G_N|$ and
\begin{equation}
 \frac{\alpha_{k,r}(D)}{\gamma_{k,r}(D)}
 =\prod_{\ell=k-r-1}^{k-2}(D+2^\ell)=H_{k,r}(D).
 \label{eqS:H-ratio}
\end{equation}
This proves Eq.~\eqref{eqS:H}.  Its equivalent flag recurrence is
\begin{equation}
 H_{k,0}=1,\qquad
 H_{k,i}=(D+2^{k-i-1})H_{k,i-1}.
 \label{eqS:H-recurrence}
\end{equation}
The zero transverse parameter is included, so there is no subtraction of
one from the second summand.
\end{proof}

\subsection{A normalization example}
\label{app:normalization-example}

\begin{remark}[Normalization at $k=5,r=1$]
There are five ordinary defect lines
$N_i=\langle\one+e_i\rangle$.  For one such line, the overlap sum and CSS
compression above give
$\gamma=(D+4)/D$ and $\alpha=(D+4)(D+8)/D$.
Their ratio is $H=\alpha/\gamma=D+8$.  Thus the eigenvalue calculation
uses both the code compression and the overlap normalization; the overlap
factor $D+4$ alone does not give the scalar.
\end{remark}

\subsection{Proof of Lemma~\ref{lemS:parabolic}}
\label{app:proof-parabolic}

\begin{proof}
Choose a Witt decomposition
$E_k=N\oplus W\oplus N^\vee$, with
$\dim W=s=k-2r$, adapted also to the stochastic vector $w_0\in W$.  Thus
\begin{equation}
 q_k(x,w,z)=q_W(w)+2x^{\mathsf T}z\pmod4,
 \qquad \one_k=(0,w_0,0),
 \label{eqS:Witt-q}
\end{equation}
and the polar form has matrix
\begin{equation}
 J=\begin{pmatrix}0&0&I_r\\0&J_W&0\\I_r&0&0\end{pmatrix}.
 \label{eqS:Witt-J}
\end{equation}
Begin with a general upper-block element that fixes the induced action on
$Q_N\simeq W$,
\begin{equation}
 g=\begin{pmatrix}A&B&E\\0&I_s&Y\\0&0&Z\end{pmatrix}.
 \label{eqS:parabolic-block}
\end{equation}
The $(1,3)$ and $(2,3)$ blocks of $g^{\mathsf T}Jg=J$ give
\begin{equation}
 Z=A^{-\mathsf T},\qquad B=AY^{\mathsf T}J_W.
 \label{eqS:parabolic-polar-one}
\end{equation}
Writing $E=AC$, the $(3,3)$ block gives the remaining polar equation
\begin{equation}
 C+C^{\mathsf T}=Y^{\mathsf T}J_WY.
 \label{eqS:C-equation}
\end{equation}
Fixing $(0,w_0,0)$ is equivalent, using
Eq.~\eqref{eqS:parabolic-polar-one}, to
\begin{equation}
 Y^{\mathsf T}J_Ww_0=0.
 \label{eqS:parabolic-stochastic}
\end{equation}
If $y_i$ is the $i$th column of $Y$, preservation of $q_k$ on the $i$th
basis vector of $N^\vee$ is exactly
\begin{equation}
 C_{ii}=q_W(y_i)/2\quad\text{in }\F_2.
 \label{eqS:parabolic-diagonal}
\end{equation}
This is meaningful: $w_0$ is the Wu vector of $W$, so
$y_i\perp w_0$ implies $q_W(y_i)\equiv b_W(y_i,w_0)=0\pmod2$.
It also makes the diagonal of
Eq.~\eqref{eqS:C-equation} consistent.

Equations~\eqref{eqS:parabolic-polar-one}--\eqref{eqS:parabolic-diagonal}\noeqref{eqS:C-equation,eqS:parabolic-stochastic}
give a bijection between kernel elements and the independent data
\begin{equation}
 A\in\mathrm{GL}(r,2),\qquad
 Y\in\operatorname{Hom}(\F_2^r,w_0^\perp),\qquad
 (C_{ij})_{i<j}\in\F_2^{\binom r2}.
 \label{eqS:kernel-data}
\end{equation}
Indeed, the diagonal of $C$ is fixed by
Eq.~\eqref{eqS:parabolic-diagonal}, its lower triangle is then fixed by
Eq.~\eqref{eqS:C-equation}, and $B,E,Z$ are fixed by the displayed block
relations.  Conversely, substitution proves polar and stochastic
preservation.  Once the polar form is preserved, the difference
$q_k(gv)-q_k(v)$ is linear in $v$.  On $(x,0,0)$ and $(0,w,0)$ the
refinement is unchanged automatically, because their images $(Ax,0,0)$
and $(Bw,w,0)$ have zero $N^\vee$ coordinate; and
Eq.~\eqref{eqS:parabolic-diagonal} checks $(0,0,e_i)$ for every basis
vector $e_i$ of $N^\vee$.  Thus the full Witt basis has been checked, and
the complete $\mathbb Z_4$ refinement is preserved.

Since
\begin{equation}
 |\mathrm{GL}(r,2)|
 =2^{r(r-1)/2}\prod_{j=1}^{r}(2^j-1),
 \end{equation}
the number of triples is
\begin{align}
 |\mathrm{GL}(r,2)|
 2^{r(s-1)}2^{r(r-1)/2}
 &=2^{r(r-1)+r(s-1)}\prod_{j=1}^{r}(2^j-1)\notag\\
 &=2^{r(k-r-2)}\prod_{j=1}^{r}(2^j-1),
 \end{align}
which proves Eq.~\eqref{eqS:kernel}.  Surjectivity is also explicit: every
stochastic $q$-isometry $h$ of $W\simeq N^\perp/N$ has the lift
$\operatorname{diag}(I_r,h,I_r)$, which preserves
Eq.~\eqref{eqS:Witt-q}; equivalently this is Lemma III.11 of
Ref.~\cite{MontealegreGross2025}.  When $r=0$, all three parameter spaces
are trivial and the kernel has order one.  When $s=1$,
$w_0^\perp=0$, so $Y=0$ and the same formulas apply without deleting an
unexamined block.  These include every ordinary terminal Witt-index case.
Finally,
$|\cG_{k,r}|=|O_k^{\rm st}(2)|/|P_r|$ and
$|P_r|=|G_N||\ker(P_r\to G_N)|$, proving
Eq.~\eqref{eqS:orbit-product}.
\end{proof}

\subsection{Proof of Proposition~\ref{propS:eigenvalues}}
\label{app:proof-eigenvalues}

\begin{proof}
Lemma~\ref{lemS:flags} gives, with the nonorthogonal CSS normalization
already included,
\begin{equation}
 \theta_{k,r}(D)=|\cG_{k,r}|\,|G_N|
 \prod_{\ell=k-r-1}^{k-2}(D+2^\ell).
 \label{eqS:theta-before}
\end{equation}
Insert Eq.~\eqref{eqS:orbit-product} and use
\begin{equation}
 \prod_{j=0}^{r-1}(2^{n+3-k}+2^{1-j})
 =2^{r(3-k)}
 \prod_{\ell=k-r-1}^{k-2}(D+2^\ell).
 \label{eqS:product-conversion}
\end{equation}
This gives Eq.~\eqref{eqS:theta}.  Equations
\eqref{eqS:gram-frame} and \eqref{eqS:moment-GNW} multiply the eigenvalue by
$N_n/Z_{n,k}$, proving Eq.~\eqref{eqS:lambda-finite}.
\end{proof}

\section{Exceptional all-ones sectors}
\label{app:exceptional}

\subsection{Proof of Lemma~\ref{lemS:shortening}}
\label{app:proof-shortening}

\begin{proof}
\smallskip\noindent\emph{Fibre classification.---}
Write $W=\F_2^{k-1}$ and let $\jmath(x)=(x,0)$.  For $x\in W$ set
\begin{equation}
 \epsilon(x)=x\cdot w,\qquad \iota(x)=(x,\epsilon(x))\in u^\perp.
\end{equation}
On one binary column define
\begin{equation}
 \mathcal U_1|x\rangle
 =|\epsilon(x)\rangle\otimes2^{-1/2}
  (|\iota(x)\rangle+|\iota(x)+u\rangle).
 \label{eqS:shortening-one-column}
\end{equation}
The map
$x\mapsto(\epsilon(x),[\iota(x)]_{\langle u\rangle})$ is a bijection from
$W$ to $\F_2\times(u^\perp/\langle u\rangle)$: replacing $x$ by $x+w$
toggles the first entry and adds $u$ to $\iota(x)$.  Hence
$\mathcal U_1$ is unitary.  Tensoring Eq.~\eqref{eqS:shortening-one-column}
over the $n$ physical columns and regrouping the first factors gives
$\mathcal U_n$ in Eq.~\eqref{eqS:shortening-unitary}.

We next define shortening on every relation matrix unit.  Proposition 4.17
of Ref.~\cite{GrossNezamiWalter2021} gives unique defect data
\begin{equation}
 T\longleftrightarrow(K_L,K_R,J_T),\qquad
 J_T:K_R^\perp/K_R\longrightarrow K_L^\perp/K_L.
 \label{eqS:relation-defect-data}
\end{equation}
Equivalently, Eq. (4.17) there reconstructs the matrix unit as
\begin{equation}
 T=\{(x,y)\in K_L^\perp\oplus K_R^\perp:
             [x]_{K_L}=J_T[y]_{K_R}\}.
 \label{eqS:relation-from-defect-data}
\end{equation}
Set
\begin{equation}
 \begin{gathered}
 \widehat K_\alpha=K_\alpha+\langle u\rangle,\\
 M_\alpha:=\{x\in W:\jmath(x)\in\widehat K_\alpha\},
 \quad\alpha=L,R.
 \end{gathered}
 \label{eqS:relation-short-defect}
\end{equation}
$\widehat K_\alpha$ is isotropic because every defect lies in $u^\perp$.
Hence $M_\alpha$ is isotropic and cannot contain $w$, since
$q_k(\jmath(w))=q_{k-1}(w)=3\pmod4$; it is an ordinary shortened defect.
The isometry $J_T$ fixes the class of $u$ and hence induces
$\overline J_T:Q_{\widehat K_R}\to Q_{\widehat K_L}$.  In $Q_{M_\alpha}$
the class $\bar w=[w]$ is nonsingular, since $q_{k-1}(w)=3\pmod4$, and
\begin{equation}
 Q_{M_\alpha}=\langle\bar w\rangle\perp\bar w^\perp,
 \qquad
 \vartheta_\alpha:\bar w^\perp\xrightarrow{\simeq}Q_{\widehat K_\alpha},
 \quad[x]\longmapsto[\jmath(x)].
 \label{eqS:shortened-quotient}
\end{equation}
Define $\operatorname{sh}(T)\in\Sigma_{k-1,k-1}(2)$ to have defects
$M_L,M_R$ and quotient isometry
\begin{equation}
 J_T^-:=\id_{\langle\bar w\rangle}\perp
 \vartheta_L^{-1}\overline J_T\vartheta_R.
 \label{eqS:relation-shortening-map}
\end{equation}
Equations~\eqref{eqS:relation-short-defect}--\eqref{eqS:relation-shortening-map}\noeqref{eqS:shortened-quotient}
show directly that $\operatorname{sh}$ is well defined, including when
$u\in K_\alpha$.

We classify a complete fibre.  Fix
$T^-=(M_L,M_R,J^-)$ and put
$d=\dim M_L=\dim M_R$.  Let
$N_\alpha=\langle u,\jmath(M_\alpha)\rangle$.  Since $J^-$ fixes the
nonsingular Wu vector $\bar w$ (every isometry of a nondegenerate odd-type
space fixes its Wu vector), its restriction to
$\bar w^\perp$ gives an isometry
$\overline J:Q_{N_R}\to Q_{N_L}$.  There is exactly one
\emph{closed} lift: it has defects $N_L,N_R$ and quotient isometry
$\overline J$.

Every other lift is \emph{transverse}.  Independently choose hyperplanes
$K_\alpha<N_\alpha$ not containing $u$; there are $2^d$ on each side.
Use the canonical identifications
$[u]^\perp/\langle[u]\rangle\simeq Q_{N_\alpha}$ inside
$Q_{K_\alpha}$.  For fixed hyperplanes, the stochastic lifts
\begin{equation}
 J:Q_{K_R}\longrightarrow Q_{K_L},\qquad
 J[u]=[u],\qquad
 J|_{[u]^\perp/\langle[u]\rangle}=\overline J
 \label{eqS:transverse-lift}
\end{equation}
form a torsor of size $2^h$, where $h=k-2d-2$.  Indeed, after one lift
identifies the two sides, write
$Q_K=\langle u,v\rangle\perp Q_N$.  Every element in the kernel of
reduction is uniquely
\begin{align}
 g_a(u)&=u,\qquad
 g_a(x)=x+b(a,x)u\quad(x\in Q_N),\notag\\
 g_a(v)&=v+a+\frac{q(a)}2u,\qquad a\in Q_N.
 \label{eqS:shortening-unipotent}
\end{align}
An initial lift exists by choosing $v_L,v_R$ with $b(u,v_\alpha)=1$;
adding $u$ to one of them, if necessary, makes
$q(v_L)=q(v_R)$, after which
$u\mapsto u$, $v_R\mapsto v_L$, and $\overline J$ on the orthogonal
complements define it.  Indeed, $[u]$ is the Wu vector and
$b(u,v_\alpha)=1$, so
$q(v_\alpha)\equiv b(v_\alpha,v_\alpha)=b(u,v_\alpha)=1\pmod2$;
the two odd values differ by $0$ or $2$, and adding $u$ toggles between
them.  The polar form on $Q_N$ is alternating because the
Wu class $u$ has been removed, so $q(a)$ is even.  Substitution in
$q(x+y)=q(x)+q(y)+2b(x,y)$ proves that every $g_a$ preserves $q$, and the
polar equations force this form.  Thus the transverse part of the fibre has
\begin{equation}
 2^d2^d2^h=2^{k-2}
 \label{eqS:shortening-fiber-count}
\end{equation}
elements.

\smallskip\noindent\emph{Disjoint exhaustion.---}
This list is disjoint and exhaustive.  If $u\in K_\alpha$, then
Eq.~\eqref{eqS:relation-short-defect} forces $K_\alpha=N_\alpha$, giving
the unique closed lift.  If $u\notin K_\alpha$, then $K_\alpha$ is exactly
a hyperplane of $N_\alpha$ missing $u$, and its quotient isometry is one of
Eq.~\eqref{eqS:transverse-lift}.  Uniqueness of the defect data
\eqref{eqS:relation-defect-data} also gives
$u\in K_L$ if and only if $u\in K_R$, so no mixed closed--transverse case
occurs.  It proves both disjointness and exhaustion;
there is no unclassified lower-filtration remainder.  Thus the no-double-count
claim follows from uniqueness of the complete defect data, not merely from
agreement of the total cardinality.

\smallskip\noindent\emph{Represented weights.---}
It remains to compute the operator weights.  Lemma 4.23 of
Ref.~\cite{GrossNezamiWalter2021} gives
$r(T)=2^eV_T$ for a relation of defect $e$, with $V_T$ the normalized
partial isometry between its CSS coset bases.  Therefore
$R_n(T)=D^eV_T^{\otimes n}$.  For a transverse hyperplane
$K_\alpha<N_\alpha$,
\begin{equation}
 |N_\alpha,[x]\rangle=2^{-1/2}
 (|K_\alpha,[x]\rangle+|K_\alpha,[x+u]\rangle).
 \label{eqS:coset-hyperplane}
\end{equation}
Because Eq.~\eqref{eqS:transverse-lift} fixes $[u]$, compression by $P_u$
makes every transverse lift induce the same normalized partial isometry
$V_{\overline J}$ from $C_{N_R}$ to $C_{N_L}$.  The closed lift induces it
directly.  Its weight is $D^{d+1}$, whereas each of the $2^{k-2}$
transverse lifts has weight $D^d$.

\smallskip\noindent\emph{Operator and filtration consequences.---}
On the shortened side, Eq.~\eqref{eqS:shortening-one-column} sends the CSS
coset basis tensorwise according to
\begin{equation}
 |M,[x]\rangle\longmapsto
 |\epsilon(x)\rangle\otimes
 |\langle u,\jmath(M)\rangle,[\iota(x)]\rangle.
 \label{eqS:shortening-coset-map}
\end{equation}
Since $J^-=\id_{\langle\bar w\rangle}\perp\overline J$, the normalized
partial isometry of $T^-$ becomes $I_D\otimes V_{\overline J}$ after all
$n$ columns.  Fibre by fibre,
\begin{equation}
 \begin{aligned}
 &I_D\otimes\!\left[\sum_{T:\operatorname{sh}(T)=T^-}
 P_uR_n(T)P_u\right]_{C_{\langle u\rangle}}\\
 &\quad=(D^{d+1}+2^{k-2}D^d)I_D\otimes V_{\overline J}\\
 &\quad=(D+2^{k-2})\mathcal U_nR_n(T^-)\mathcal U_n^\dagger.
 \end{aligned}
 \label{eqS:shortening-fibre-operator}
\end{equation}
Summing Eq.~\eqref{eqS:shortening-fibre-operator} over $T^-$ proves
Eq.~\eqref{eqS:shortening-intertwiner}.  The numerical factor has now been
assigned once and only once: $D^{d+1}$ is the single closed weight,
$2^{k-2}D^d$ is the total transverse weight, and the common $D^d$ is exactly
the weight of the shortened relation $R_n(T^-)$.  Because this is an operator
identity after physical compression on the entire all-ones code, no kernel,
stabilizer, or further multiplicity factor can alter the spectrum.

For any ordinary defect $M<W$, the same coset formula gives
\begin{equation}
 \mathcal U_n C_M
 =\mathbb C^D\otimes C_{\langle u,\jmath(M)\rangle}.
 \label{eqS:code-shortening}
\end{equation}
The map $M\mapsto\langle u,\jmath(M)\rangle$ is an
inclusion-preserving bijection from the ordinary defect poset at $k-1$ to
the exceptional defect poset at $k$, increasing dimension by one.  It maps
the maximal ordinary defect in each orthogonal type to the maximal
exceptional defect of the paired type; therefore the terminal grade is
included and no boundary sector is lost or repeated.  Hence
\begin{equation}
 \mathcal U_n\mathcal F^{\Ord}_{k-1,r-1}
 =\mathbb C^D\otimes\mathcal F^{\Exc}_{k,r},\qquad
 \mathcal U_n\mathcal K_{k-1,r-1}
 =\mathbb C^D\otimes\mathcal L_{k,r}.
 \label{eqS:shortening-filtrations}
\end{equation}
This also accounts for every filtration term: shortened defects of
dimension at least $r$ lie in the deeper spaces on both sides, whereas
smaller defects remain and are matched fibrewise by
Eq.~\eqref{eqS:shortening-fibre-operator}.

Finally,
\begin{equation}
 \frac{Z_{n,k}}{Z_{n,k-1}}=D+2^{k-2}.
 \label{eqS:Z-shortening}
\end{equation}
Thus Eq.~\eqref{eqS:shortening-intertwiner}, normalized by $N_n/Z_{n,k}$,
is an exact unitary equivalence of the frame operators on the associated
grades.  The left side has every target eigenvalue repeated $D$ times,
whereas the right side has the shortened multiplicity.  The Gram--frame
correspondence~\eqref{eqS:gram-frame} now gives both identities in
Eq.~\eqref{eqS:shortening-spectrum}.
\end{proof}

\subsection{Quadratic-refinement and endpoint checks}
\label{app:quadratic-checks}
All binary classifications and lifts above retain the full
$\mathbb Z_4$-valued refinement $q_k$ of Eq.~\eqref{eqS:quadratic}, rather
than only its polar form.  Stochastic relations, their defects, and
$O_k^{\rm st}(2)$ are defined by preservation or vanishing of $q_k$ in
Section~\ref{sec:relations}.  The weighted-defect argument verifies that the canonical quotient
map preserves $q_k$ and checks the mod-four hypothesis in Witt extension
after Eq.~\eqref{eqS:weighted-canonical-map}.  The orbit-tail quotient and
classification use Proposition II.13 and Lemma III.4 of
Ref.~\cite{MontealegreGross2025} with this same refinement and with Wu-line
containment retained.  The parabolic count begins
with the quadratic Witt form~\eqref{eqS:Witt-q}, imposes the diagonal
refinement equation~\eqref{eqS:parabolic-diagonal}, and checks the full Witt
basis before counting.  In the shortening proof,
$q_k(\one_k)=0$ for $4\mid k$, the shortened Wu class has the stated odd
quadratic value, and Eq.~\eqref{eqS:shortening-unipotent} is checked on the
complete transverse torsor using the quadratic addition law.  The signs in
Eqs.~\eqref{eqS:type}, \eqref{eqS:ranges}, and
\eqref{eqS:exceptional-ranges}, together with the zero-dimensional
conventions following Eq.~\eqref{eqS:orthogonal-order}, retain the two binary
orthogonal towers through their terminal layers.  Finally, the CSS vectors
in Eq.~\eqref{eqS:CSS-coset} are explicit equal-amplitude
computational-basis sums, while Lemma~\ref{lemS:mixed-moment} is an identity
of rank-one projectors and Eq.~\eqref{eqS:Rrelation} fixes the relation
matrix units with unit coefficients.  Consequently arbitrary Pauli-lift
phases enter neither the CSS overlaps nor those represented matrix units.

\section{Eigenvalue multiplicities}
\label{app:multiplicity}

\subsection{Proof of Lemma~\ref{lemS:invariant-code}}
\label{app:proof-invariant-code}

We use the fixed-space enumerators defined in
Section~\ref{sec:multiplicity}.
We now connect these enumerators to the Clifford code filtration.  Let $C_N$
be the code representation space attached to a defect $N$ and let
$Q_N=N^\perp/N$.  In the stable range, the code model of
Ref.~\cite{MontealegreGross2025} identifies
\begin{equation}
 C_N\simeq\mathbb C[\operatorname{Hom}(\F_2^n,Q_N)].
 \label{eqS:code-model}
\end{equation}
The action is the finite permutation action
$g\cdot F=g\circ F$ on $\operatorname{Hom}(\F_2^n,Q_N)$.  If
$f=\dim\Fix_{Q_N}(g)$, then a homomorphism is fixed precisely when its image
lies in $\Fix_{Q_N}(g)$, so the summand contributed by $g$ is
$2^{nf}=D^f$.  Burnside's lemma therefore gives the exact average
\begin{equation}
 \begin{aligned}
 \dim C_N^{G_N}
 &=\bigl|G_N\backslash\operatorname{Hom}(\F_2^n,Q_N)\bigr|\\
 &=\frac1{|G_N|}\sum_{g\in G_N}D^{\dim\Fix_{Q_N}(g)}.
 \end{aligned}
 \label{eqS:Burnside-code}
\end{equation}
Thus every term in the average counts the fixed homomorphisms of one
specified group element; no formal character multiplicity is being inserted.

\begin{proof}
For odd $k$, the quotient has a fixed nonsingular Wu line orthogonal to a
$2m$-dimensional $O^\delta$ module.  Hence every fixed-space dimension is
$1+f(g)$, and Burnside gives Eq.~\eqref{eqS:a-odd}.

For $k\equiv2\pmod4$, choose a Witt basis that separates the stochastic
two-plane from a $(2m-2)$-dimensional symplectic module.  Each
$h\in\Sp(2m-2,2)$ has two stochastic lifts.  Solving the fixed-vector
equations in that two-plane shows that their fixed-space contributions are
$D^{f(h)+2}$ and $D^{f(h)+1}$.  Averaging the two lifts gives
$D(D+1)D^{f(h)}/2$, proving Eq.~\eqref{eqS:a-two}.

For $4\mid k$, the quotient stochastic group is
$2^{2m-2}\rtimes O^\delta(2m-2,2)$.  Write an element as $(u,h)$, with
$u$ in the natural orthogonal module $W$ and
$f=\dim\Fix_W(h)$.  The fixed-vector equations in a stochastic Witt basis
have dimension $f+2$ when $u\in\Fix(h)^\perp$ and dimension $f$ otherwise.
Exactly $2^{2m-2-f}$ of the $2^{2m-2}$ vectors $u$ satisfy the first
condition.  Thus the average over $u$, for fixed $h$, is
\begin{equation}
 D^f+(D^2-1)(D/2)^f.
 \label{eqS:semidirect-average}
\end{equation}
Averaging Eq.~\eqref{eqS:semidirect-average} over $h$ proves
Eq.~\eqref{eqS:a-zero}.
\end{proof}

\subsection{Proof of Proposition~\ref{propS:multiplicity}}
\label{app:proof-multiplicity}

Proposition~\ref{propS:filtration-bridge} now turns these Burnside orbit
counts into physical multiplicities.  The quotient-orbit bijection
\eqref{eqS:defect-orbit-tail-bijection} identifies $a_{s-2}$ with the tail
from larger defects inside $a_s$, after excluding the ordinary and, when
present, all-ones heads.  The adjacent difference removes that tail exactly
and leaves each head once.  For $4\nmid k$ the sole head is
the ordinary multiplicity $a_s-a_{s-2}$; for $4\mid k$ the same drop is the
disjoint sum of ordinary grade $r$ and exceptional grade $r+1$.  These are
equalities of represented physical subspaces, so an abstract module killed by
physical compression is never counted as a multiplicity.

\begin{proof}
We give the odd tower continuously from its action to its product formula.
For $s=2m+1$, Eq.~\eqref{eqS:Burnside-code} and the fixed nonsingular Wu line
give $a_{2m+1}=DB_m^{O,\delta}(D)$.  The physical adjacent-tail identity then
gives
\begin{equation}
 \begin{aligned}
 M^{\Ord}_{k,r}(D)
 &=a_{2m+1}(D)-a_{2m-1}(D)\\
 &=D\,[B_m^{O,\delta}(D)-B_{m-1}^{O,\delta}(D)]\\
 &=\frac{D(D-\delta2^{m-1})(D+2^{2m-1})
  \prod_{j=0}^{m-2}(D^2-4^j)}{|O^\delta(2m,2)|},
 \end{aligned}
\end{equation}
where the last equality is the exact increment
\eqref{eqS:Oincrement}.  This is Eq.~\eqref{eqS:M-odd}.  At $m=0$ the plus
tower ends with $a_1=D$ and the zero-by-nonexistence value $a_{-1}=0$, hence
$M=D$; the corresponding minus endpoint is absent.

The other towers use exactly the same action--Burnside--adjacent-difference
mechanism.  For $k\equiv2\pmod4$, the corresponding difference is
\begin{equation}
 \frac{D(D+1)}2[B_{m-1}^{\Sp}(D)-B_{m-2}^{\Sp}(D)].
\end{equation}
For $m\geq2$, insert Eq.~\eqref{eqS:Spincrement}; its factor $D-1$
combines with $D+1$ to produce the missing $j=0$ factor $D^2-1$ in
Eq.~\eqref{eqS:M-two}.  At $m=1$, the convention
$B_{-1}^{\Sp}=0$ gives the endpoint
$M^{\Ord}=D(D+1)/2$ directly.

For $4\mid k$, Eq.~\eqref{eqS:a-zero} is first differenced at $s$ and
$s-2$.  Proposition~\ref{propS:filtration-bridge} gives the exact
no-double-count identity
\begin{equation}
 a_s(D)-a_{s-2}(D)
 =M^{\Ord}_{k,r}(D)+M^{\Exc}_{k,r+1}(D).
 \label{eqS:multiplicity-split}
\end{equation}
To identify the two heads individually, use the operator-shortening
Lemma~\ref{lemS:shortening} in Section~\ref{sec:exceptional}.  Its proof uses relation operators
and code projectors, independently of the present $4\mid k$ multiplicity
formula.  Together with the odd multiplicities already established, it gives
\begin{equation}
 M^{\Exc}_{k,r+1}(D)=D^{-1}M^{\Ord}_{k-1,r}(D).
 \label{eqS:exceptional-head-by-shortening}
\end{equation}
The odd problem at $k-1$ has residual dimension
$(k-1)-2r=s-1=2m-1$; the definitions of $\delta$ agree for the paired
residues $k\equiv0,4\pmod8$ and $k-1\equiv7,3\pmod8$, respectively.
The odd case already established above therefore
gives, before inserting the closed group orders,
\begin{equation}
 D^{-1}M^{\Ord}_{k-1,r}(D)
 =B_{m-1}^{O,\delta}(D)-B_{m-2}^{O,\delta}(D).
 \label{eqS:exceptional-head-increment}
\end{equation}
Thus the first increment in the difference of Eq.~\eqref{eqS:a-zero} is
the exceptional multiplicity, rather than merely one unspecified member
of the two-head sum in Eq.~\eqref{eqS:multiplicity-split}.  Removing it
leaves the ordinary remainder
\begin{equation}
 (D^2-1)[B_{m-1}^{O,\delta}(D/2)
          -B_{m-2}^{O,\delta}(D/2)].
 \label{eqS:ordinary-remainder}
\end{equation}
Insert Eq.~\eqref{eqS:Oincrement} at $z=D/2$.  Pulling a factor $1/4$ from
each quadratic term and $1/2$ from each linear term produces the denominator
$2^{2m-2}$; multiplication by $D^2-1$ completes the product from $j=0$.
This is Eq.~\eqref{eqS:M-zero}.  At the admissible $m=1$ plus endpoint,
$a_2=D^2$ splits as the exceptional contribution $1$ and the ordinary
contribution $D^2-1$; the zero-dimensional minus tower is absent.

For clarity, the terminal conventions used in this derivation are: odd plus
has the $m=0$ value $M=D$ and odd minus has no $m=0$ sector; the symplectic
tower has the $m=1$ value $D(D+1)/2$; and, when $4\mid k$, plus type has the
$m=1$ split $(D^2-1)+1$, while the corresponding minus sector is absent.
All remaining values follow from the same two fixed-space increments
\eqref{eqS:Oincrement}--\eqref{eqS:Spincrement}.
\end{proof}

\section{Complete spectrum at every physical dimension}
\label{app:finite}

\subsection{Proof of Lemma~\ref{lemS:polynomial-continuation}}
\label{app:proof-continuation}

\begin{proof}
\emph{Step 1: actual stable-range power sums.---}
For every integer $p\geq1$, repeated use of the finite relation product
\eqref{eqS:semigroup} and trace formula~\eqref{eqS:relation-trace} expands
\begin{equation}
 \Tr[A_k(D)^p]
 \label{eqS:trace-polynomial}
\end{equation}
as a finite sum of monomials $D^j$, and hence as a polynomial in $D$.  In
the stable range, physical exhaustion from
Proposition~\ref{propS:filtration-bridge} gives its power-sum identity with
the candidate spectrum.

\emph{Step 2: polynomial continuation.---}
The candidate power sum
\begin{equation}
 \widehat p_p(D):=
 \sum_{a\in\mathfrak A_k}M_a(D)\theta_a(D)^p
 \label{eqS:predicted-polynomial}
\end{equation}
is also a polynomial.  Indeed, in Eq.~\eqref{eqS:theta} one writes
$2^{n+3-k}=2^{3-k}D$; all remaining group orders depend only on $k$, and
Eqs.~\eqref{eqS:M-odd}--\eqref{eqS:M-zero}\noeqref{eqS:M-two} have no $D$-dependent
denominator.  On the exceptional branch the apparent $D^{-1}$ in
Eq.~\eqref{eqS:shortening-spectrum} cancels the explicit leading $D$ in the
odd-$(k-1)$ multiplicity, while
$\theta^{\Exc}_{k,r}=(D+2^{k-2})\theta^{\Ord}_{k-1,r-1}$ is polynomial.
At every stable value $D=2^n$, $n\geq k$, the two sides agree.  Since this
is an infinite set of values, the two polynomials are identical and therefore
agree at every physical specialization.

\emph{Step 3: integrality independent of exhaustion.---}
At a physical specialization, every $M_a(D)$ is an integer.  This assertion
does not use the stable-range filtration theorem.  For
$D=2^n$, each $B_G(D)$ is, by Burnside, the number of orbits of $G$ on
$\operatorname{Hom}(\F_2^n,V)$.  The odd and $k\equiv2\pmod4$ formulas
are respectively
\begin{equation}
 \begin{gathered}
 D\,[B_m^{O,\delta}(D)-B_{m-1}^{O,\delta}(D)],\\
 \frac{D(D+1)}2\,[B_{m-1}^{\Sp}(D)-B_{m-2}^{\Sp}(D)].
 \end{gathered}
 \label{eqS:integral-Burnside-differences}
\end{equation}
and the $4\mid k$ ordinary formula is exactly
Eq.~\eqref{eqS:ordinary-remainder}.  Here $B_G(D/2)$ is again an orbit
count because $D/2=2^{n-1}$, including $D/2=1$ when $n=1$.  Thus all
three are integers directly.  On the exceptional branch, division by
$D$ cancels the explicit leading $D$ in the odd multiplicity, leaving an
integer difference of orthogonal Burnside counts.

\emph{Step 4: admissibility, zeros, and nonnegativity.---}
Nonnegativity can be read factor by factor from
Eqs.~\eqref{eqS:M-odd}--\eqref{eqS:M-zero}\noeqref{eqS:M-two}:
within the available rank every factor is nonnegative; once the quotient
rank is too large, the first plus-tower over-rank layer is killed by
$D-2^{m-1}$, and candidates with still larger residual dimension are killed
because $\prod_j(D^2-4^j)$ contains $D^2-4^n$.
The endpoint conventions in Eqs.~\eqref{eqS:M-odd} and
\eqref{eqS:orthogonal-order} cover the remaining zero-dimensional cases,
and formally absent candidates already have $M_a=0$ by definition.  This is
not a blanket ``over-rank implies zero'' rule: the first layers outside the
stable range can survive and are intentionally retained; only the stated
factor analysis annihilates the candidates with sufficiently large residual
dimension or an inadmissible endpoint before the spectral sum is formed.
The exceptional multiplicities inherit the same nonnegativity from
$M^{\Exc}_{k,r}=D^{-1}M^{\Ord}_{k-1,r-1}$.

\emph{Step 5: positivity of every surviving atom.---}
Every candidate $\theta_a(D)$ with positive multiplicity is strictly
positive by Eqs.~\eqref{eqS:theta} and
\eqref{eqS:shortening-spectrum}.  Thus Steps 3--5 establish a positive
candidate measure without using the conclusion of this lemma.

\emph{Step 6: identification, including coincidences.---}
The actual operator $A_k(D)=Z_{n,k}\rho_{n,k}$ is positive.  We have now
constructed a second finite positive atomic measure whose positive power
moments all equal those of its nonzero spectrum.  The two measures coincide:
equivalently, the rational functions
\begin{equation}
 \Tr\!\left[A_k(D)(I-zA_k(D))^{-1}\right]
 \quad\text{and}\quad
 \sum_{a\in\mathfrak A_k}
 \frac{M_a(D)\theta_a(D)}{1-z\theta_a(D)}
 \label{eqS:resolvents}
\end{equation}
have the same Taylor series near zero and hence the same poles and combined
residues, including when two displayed sector eigenvalues coincide.  This
determines every nonzero eigenvalue of $A_k(D)$ and its total multiplicity.

\emph{Step 7: the two zero eigenspaces.---}
Equations~\eqref{eqS:gram-frame} and \eqref{eqS:moment-GNW} give
\begin{equation}
 S_{n,k}=\frac{N_n}{Z_{n,k}}A_k(D),
 \label{eqS:frame-relation-scaling}
\end{equation}
and polar decomposition identifies the nonzero spectrum of $H_{n,k}$ with
that of $S_{n,k}$.  Thus $A_k(D)$ and $H_{n,k}$ have the same rank and the
same nonzero multiplicities, while their corresponding eigenvalues differ
by the factor $N_n/Z_{n,k}$.  Their ambient dimensions are $D^k$ and $N_n$,
respectively.  Subtracting their common rank from these dimensions proves
Eqs.~\eqref{eqS:A-zero-multiplicity} and
\eqref{eqS:H-zero-multiplicity} separately.
\end{proof}

\subsection{Proof of Theorem~\ref{thmS:finite-spectrum}}
\label{app:proof-finite-spectrum}

\begin{proof}
For fixed $k$ and stable $n\geq k$, the multiplicity-free Clifford code
decomposition in Proposition~\ref{propS:filtration-bridge}, together with
Propositions~\ref{propS:eigenvalues} and \ref{propS:multiplicity} and
Lemma~\ref{lemS:shortening}, exhausts the stabilizer-tensor span.  Hence the
theorem holds in that range.
Lemma~\ref{lemS:polynomial-continuation} removes the stable-range
restriction and supplies the zero eigenspaces.  Scaling by $N_n/Z_{n,k}$ and
using Lemma~\ref{lemS:PGM} proves Eq.~\eqref{eqS:finite-success}.
\end{proof}

\section{Fixed-offset limit and normalization}
\label{app:limit}

\subsection{Proof of Lemma~\ref{lemS:limit-bounds}}
\label{app:proof-limit-bounds}

\begin{proof}
We display the finite cancellations for each residue class.  Put
\begin{equation}
 \begin{gathered}
 A_n:=\prod_{j=1}^n(1+2^{-j}),\\
 Q_j:=\prod_{i=1}^j(1-4^{-i}),\qquad Q_0:=1,\\
 T_{n,m}:=\prod_{i=n-m+2}^{n}(1-4^{-i}).
 \end{gathered}
 \label{eqS:finite-products}
\end{equation}
with an empty product equal to one, and set
$u:=2r-c$ and $E(u):=(u^2+3u+4)/2$.  Direct insertion of
Eqs.~\eqref{eqS:M-odd}--\eqref{eqS:M-zero}\noeqref{eqS:M-two} and the group orders gives the
following exact identities in the bulk:
For odd $k$, with $s=2m+1$,
\begin{equation}
 \begin{aligned}
 \frac{M^{\Ord}_{k,r}}{N_n}
 &=\frac{2^{-E(u)}(1+2^{-u-2})}{A_n}\\
 &\quad\times\frac{1-\delta2^{m-1-n}}{1-\delta2^{-m}}
   \frac{T_{n,m}}{Q_{m-1}}.
 \end{aligned}
 \label{eqS:finite-cancel-odd}
\end{equation}
For $k\equiv2\pmod4$, with $s=2m$,
\begin{equation}
 \frac{M^{\Ord}_{k,r}}{N_n}
 =\frac{2^{-E(u)}(1+2^{-u-2})}{A_n}
   \frac{T_{n,m}}{Q_{m-1}}.
 \label{eqS:finite-cancel-two}
\end{equation}
For $4\mid k$, with $s=2m$,
\begin{equation}
 \begin{aligned}
 \frac{M^{\Ord}_{k,r}}{N_n}
 &=\frac{2^{-E(u)}(1+2^{-u-2})}{A_n}\\
 &\quad\times\frac{1-\delta2^{m-1-n}}{1-\delta2^{1-m}}
   \frac{T_{n,m}}{Q_{m-2}}.
 \end{aligned}
 \label{eqS:finite-cancel-zero}
\end{equation}
The first formula is used for $m\geq1$ and the last for $m\geq2$; the
endpoint conventions in Section~\ref{sec:multiplicity} cover smaller $m$.  For fixed $c,r$,
these endpoints never occur once $n$ is large.  Then $m\to\infty$,
$T_{n,m}\to1$: indeed,
\begin{equation}
 0\leq1-T_{n,m}
 \leq\sum_{i=n-m+2}^{n}4^{-i}=o(1).
 \label{eqS:T-tail-bound}
\end{equation}
Both displayed sign ratios tend to one,
$A_n\to A$, and $Q_{m-1},Q_{m-2}\to Q$.  Each of
Eqs.~\eqref{eqS:finite-cancel-odd}--\eqref{eqS:finite-cancel-zero}\noeqref{eqS:finite-cancel-two}
therefore tends to $w_r(c)$.

The eigenvalue cancellation is equally explicit.  Define
\begin{equation}
 C_k:=
 \begin{cases}
 2^{-1},&k=1,\\
 (1-\delta2^{-a})Q_{a-1},&k=2a+1\geq3,\\
 Q_{a-1},&k=2a\equiv2\pmod4,\\
 (1-\delta2^{1-a})Q_{a-2},&k=2a,\ 4\mid k.
 \end{cases}
 \label{eqS:Ck}
\end{equation}
The three order formulas reduce exactly to
\begin{equation}
 |O_k^{\rm st}(2)|=2^{(k^2-3k+4)/2}C_k.
 \label{eqS:order-cancellation}
\end{equation}
Also, with
$P_{n,c}:=\prod_{h=2-c}^{n}(1+2^{-h})$,
\begin{equation}
 Z_{n,k}=2^{nk}P_{n,c},\qquad \frac{A_n}{P_{n,c}}=R_c.
 \label{eqS:Z-fixed-offset-product}
\end{equation}
Using $N_n=2^{(n^2+3n)/2}A_n$ in
Eq.~\eqref{eqS:lambda-finite} at $r=0$ now gives the exact identity
\begin{equation}
 \lambda^{\Ord}_{k,0}(n)
 =2^{(c-1)(c-2)/2+1}R_cC_k.
 \label{eqS:lambda0-finite-cancellation}
\end{equation}
Since $C_k\to Q$, this tends to $\lambda_0(c)$.  Finally,
\begin{equation}
 \frac{\lambda^{\Ord}_{k,r}(n)}{\lambda^{\Ord}_{k,0}(n)}
 =F_r(2^{3-c})
 \label{eqS:eigenvalue-exact-ratio}
\end{equation}
at every finite $n$, proving the second limit without another asymptotic
cancellation.

For the uniform estimates, extend every ordinary summand by zero when $r$ is
not admissible, and put
\begin{equation}
 \begin{gathered}
 B:=\prod_{j=1}^{\infty}(1-2^{-j})=\frac QA,\\
 u:=2r-c,\qquad s:=k-2r=n-u.
 \end{gathered}
 \label{eqS:B-u}
\end{equation}
Let
\begin{equation}
 r_0(c):=\max\left\{0,\left\lceil\frac{c-2}{2}\right\rceil\right\}.
\end{equation}
For $r\geq r_0(c)$ one has $s-2\leq n$.  Hence every additive power in the
multiplicity formulas is at most $D$:
\begin{equation}
 \begin{gathered}
 |D-\delta2^{m-1}|\leq2D,\qquad D+2^{s-2}\leq2D,\\
 0\leq D^2-4^j\leq D^2.
 \end{gathered}
 \label{eqS:multiplicity-factor-bounds}
\end{equation}
After extracting their leading powers of two, the group-order denominators
obey
\begin{align}
 |O^\delta(2m,2)|
 &\geq \frac Q2\,2^{2m^2-m+1},
 &&m\geq1,\label{eqS:O-lower}\\
 2|\Sp(2m-2,2)|
 &\geq Q\,2^{2m^2-3m+2},
 &&m\geq1,\label{eqS:Sp-lower}\\
 2^{2m-2}|O^\delta(2m-2,2)|
 &\geq \frac Q2\,2^{2m^2-3m+2},
 &&m\geq2.\label{eqS:Ozero-lower}
\end{align}
Indeed, the remaining finite initial products of $(1-4^{-j})$ are bounded
below by $Q$, and $1-\delta2^{-m}\geq1/2$ in every admissible orthogonal
case.

Writing $g(s):=(s^2-3s+4)/2$, the leading denominator exponent in all three
cases is $g(s)$, while the numerator is at most $4D^s$.  Since
$N_n\geq2^{(n^2+3n)/2}$, this gives
\begin{align}
 \frac{M^{\Ord}_{k,r}}{N_n}
 &\leq\frac8Q\,2^{ns-g(s)-(n^2+3n)/2}\notag\\
 &=\frac8Q\,2^{-(u^2+3u+4)/2}.
 \label{eqS:Mbound}
\end{align}
The two endpoints not covered by the group-order bounds are nonzero but obey
the same estimate after enlarging its constant by at most two.  Namely, at
the odd plus-type endpoint $s=1$ one has $M^{\Ord}_{k,r}=D$, and at the
$4\mid k$ plus-type endpoint $s=2$ one has
$M^{\Ord}_{k,r}=D^2-1$; substitute $u=n-1$ and $u=n-2$, respectively, in
Eq.~\eqref{eqS:Mbound}.

For the eigenvalues, the group orders also give
\begin{equation}
 |O_k^{\rm st}(2)|\leq2^{(k^2-3k+4)/2+1}.
\end{equation}
Together with
$N_n\leq A\,2^{(n^2+3n)/2}$ and $Z_{n,k}\geq D^k$, this implies
\begin{equation}
 \lambda^{\Ord}_{k,0}(n)
 \leq L_c:=2A\,2^{(c^2-3c+4)/2}.
 \label{eqS:lambda0-uniform}
\end{equation}
The correct uniform denominator bound in $F_r$ is
\begin{equation}
 \prod_{j=1}^{r}(2^j-1)
 =2^{r(r+1)/2}\prod_{j=1}^{r}(1-2^{-j})
 \geq\frac QA\,2^{r(r+1)/2}.
 \label{eqS:Fbounds}
\end{equation}
With $x_c=2^{3-c}$, Eq.~\eqref{eq:Fr} therefore yields
\begin{align}
 F_r(x_c)
 &\leq\frac AQ\,2^{(r^2-3r)/2}(x_c+2)^r,\notag\\
 \lambda^{\Ord}_{k,r}(n)
 &\leq\frac{AL_c}{Q}\,2^{r^2/2+\eta_cr},
 \qquad
 \eta_c:=\log_2(x_c+2)-\frac32.
 \label{eqS:lambdabound}
\end{align}
Since
\begin{equation}
 -\frac{u^2+3u+4}{2}
 =-2r^2+(2c-3)r-\frac{c^2-3c+4}{2},
\end{equation}
Eqs.~\eqref{eqS:Mbound} and \eqref{eqS:lambdabound} give
\begin{align}
 \frac{M^{\Ord}_{k,r}}{N_n}\sqrt{\lambda^{\Ord}_{k,r}}
 &\leq K_c\,2^{-7r^2/4+(2c-3+\eta_c/2)r},\notag\\
 \frac{M^{\Ord}_{k,r}}{N_n}\lambda^{\Ord}_{k,r}
 &\leq K'_c\,2^{-3r^2/2+(2c-3+\eta_c)r}.
\end{align}
A negative quadratic dominates either linear term, proving
Eqs.~\eqref{eqS:amplitude-majorant} and
\eqref{eqS:moment-majorant} for $r\geq r_0(c)$.  The finitely many
$r<r_0(c)$ are uniformly bounded for all sufficiently large $n$ by their
termwise convergence; enlarging $C_c$ incorporates them.

To make the moving-boundary quantifiers explicit, the uniformity check is:
\begin{itemize}
\item Eqs.~\eqref{eqS:finite-cancel-odd}--\eqref{eqS:finite-cancel-zero}\noeqref{eqS:finite-cancel-two}
      treat all three residue classes and retain both orthogonal signs;
\item inadmissible ordinary summands are extended by zero before any bound is
      applied, including candidates that disappear at a physical
      specialization;
\item Eqs.~\eqref{eqS:O-lower}--\eqref{eqS:Ozero-lower}\noeqref{eqS:Sp-lower} are sign-uniform,
      and every denominator stays positive in its stated admissible range;
\item the terminal multiplicities at $s=1$ and $s=2$ are evaluated directly
      in the endpoint calculation following Eq.~\eqref{eqS:Mbound}, while
      their eigenvalues use Eqs.~\eqref{eqS:lambda0-uniform} and
      \eqref{eqS:lambdabound}, rather than a bulk denominator estimate;
\item the bulk argument is used only for $r\geq r_0(c)$, while the finite
      head $r<r_0(c)$ is controlled by termwise convergence; and
\item the exceptional branch is not absorbed into these ordinary estimates:
      its complete amplitude obeys the separate exact bound
      \eqref{eqS:exceptional-bound}, and its first moment is evaluated
      separately in Eq.~\eqref{eqS:exceptional-first-moment} below.
\end{itemize}
Consequently, for fixed $c$ and all sufficiently large $n$, one constant
$C_c$ works uniformly in every admissible $r$, residue class, orthogonal
sign, and distance to the terminal layer.  Coincident displayed eigenvalues
do not affect the estimate, which is proved sector by sector; the resolvent
identity~\eqref{eqS:resolvents} is used only to combine their multiplicities
when identifying the spectrum.  Disappearing candidates already have zero
multiplicity.  Thus both majorants hold for every $r\geq0$, including every
admissible boundary sector.
\end{proof}

\subsection{Proof of Theorem~\ref{thm:curve}}
\label{app:proof-curve}

\begin{proof}
Theorem~\ref{thmS:finite-spectrum} writes the square root of the finite
success probability as the sum of ordinary and, possibly, exceptional
amplitudes.  Lemma~\ref{lemS:limit-bounds} and dominated convergence give
\begin{equation}
 \lim_{n\to\infty}
 \sum_r\frac{M^{\Ord}_{k,r}}{N_n}
 \sqrt{\lambda^{\Ord}_{k,r}}
 =\sum_{r\geq0}w_r(c)\sqrt{\lambda_r(c)}.
\end{equation}
The exceptional amplitude is at most $2^{-n}$ by
Eq.~\eqref{eqS:exceptional-bound}, and therefore vanishes.  Squaring proves
Eq.~\eqref{eq:curve}.

For normalization, the exceptional first-moment contribution is also known
exactly.  When $4\mid k$, Lemma~\ref{lemS:shortening} and the fact that the
complete positive spectrum at the odd tensor power $k-1$ is ordinary give
\begin{align}
 \sum_r\frac{M^{\Exc}_{k,r}}{N_n}\lambda^{\Exc}_{k,r}
 &=D^{-1}\sum_q\frac{M^{\Ord}_{k-1,q}}{N_n}
                 \lambda^{\Ord}_{k-1,q}
 =D^{-1}.
 \label{eqS:exceptional-first-moment}
\end{align}
For $4\nmid k$ this contribution is zero.  Hence the ordinary first-moment
sum equals either $1$ or $1-D^{-1}$ and in both cases tends to one.  Applying
dominated convergence with Eq.~\eqref{eqS:moment-majorant} proves
Eq.~\eqref{eq:normalization}.
\end{proof}

\subsection{Limiting spectral mass}
\label{app:spectral-mass}

For every $n$, the total ordinary positive multiplicity obeys
\begin{equation}
 \sum_r\frac{M^{\Ord}_{k,r}}{N_n}\leq1.
\end{equation}
Fatou's lemma and Eq.~\eqref{eqS:termwise-limit} therefore give
$\sum_rw_r(c)\leq1$.  Moreover, Eq.~\eqref{eqS:Mbound} supplies a
summable majorant after a finite head, so dominated convergence gives
$\sum_rM^{\Ord}_{k,r}/N_n\to\sum_rw_r(c)$; the exceptional multiplicity
is at most $D^{-1}N_n$ by shortening.  Thus $\sum_rw_r(c)$ is the total
nonzero limiting spectral mass and the missing mass is the zero
eigenvalue.  At $c=0$, for example, the positive spectral mass is
$0.1954251444668857\ldots$.

\section{Tail asymptotics}
\label{app:tails}

\subsection{Proof of Proposition~\ref{propS:tails}}
\label{app:proof-tails}

Let
\begin{equation}
 a_r(c):=w_r(c)\sqrt{\lambda_r(c)}.
 \label{eqS:a-limit}
\end{equation}
The defining products give the exact adjacent-term ratio
\begin{equation}
 \begin{aligned}
 \frac{a_{r+1}(c)}{a_r(c)}
 &=2^{-(2u+5)}\frac{1+2^{-u-4}}{1+2^{-u-2}}\\
 &\quad\times\sqrt{\frac{2^{2r}(2^{3-c}+2^{1-r})}{2^{r+1}-1}},
 \qquad u=2r-c.
 \end{aligned}
 \label{eqS:term-ratio}
\end{equation}

\begin{proof}
For $c=-m$, Eq.~\eqref{eqS:term-ratio} gives
$a_1/a_0=O(2^{-3m/2})$, and the ratio decreases by a fixed geometric
factor with each subsequent $r$.  Thus
$\sum_ra_r=a_0[1+O(2^{-3m/2})]$.  Directly
multiplying Eqs.~\eqref{eq:weight} and
\eqref{eq:lambda0} at $r=0,c=-m$ gives
\begin{equation}
 a_0(-m)^2=\frac{2^{-(m^2+3m+4)/2}}{AQ}
 \frac{R_{-m}}{A}(1+2^{-m-2})^2.
 \label{eqS:left-leading-exact}
\end{equation}
Here $R_{-m}=\prod_{j=1}^{m+1}(1+2^{-j})$, whence
\begin{equation}
 \begin{aligned}
 \frac{R_{-m}}A&=1-2^{-m-1}+O(2^{-2m}),\\
 (1+2^{-m-2})^2&=1+2^{-m-1}+O(2^{-2m}).
 \end{aligned}
\end{equation}
The first-order corrections cancel.  Squaring the amplitude estimate and
using Eq.~\eqref{eqS:left-leading-exact} proves
Eq.~\eqref{eqS:left-asymptotic}.

For the right tail put $p=m\bmod2$,
$\mathbb Z_p:=\{u\in\mathbb Z:u\equiv p\pmod2\}$, and define
\begin{equation}
 \begin{gathered}
 \omega_u:=\frac{2^{-(u^2+3u+4)/2}(1+2^{-u-2})}{AQ},\\
 L_u:=4\,2^{u/2}-2^{-u/2}.
 \end{gathered}
 \label{eqS:omega-L}
\end{equation}
The Jacobi triple-product identity
\begin{equation}
 \sum_{\nu\in\mathbb Z}q^{\nu^2}z^\nu
 =\prod_{j=1}^{\infty}(1-q^{2j})
  (1+zq^{2j-1})(1+z^{-1}q^{2j-1})
 \label{eqS:JTP}
\end{equation}
gives a convenient generating function for the required moments.
Write $\Theta(z):=\sum_{\nu\in\mathbb Z}4^{-\nu^2}z^\nu$ and
\begin{equation}
 G_p(t):=\sum_{u\in\mathbb Z_p}\omega_u2^{tu/2}.
 \label{eq:theta-generator}
\end{equation}
The substitution $u=p+2\nu$ in this absolutely convergent sum gives
\begin{equation}
 \begin{aligned}
 G_p(t)&=\frac{2^{-(p^2+3p+4)/2+tp/2}}{AQ}\\
 &\quad\times\left[\Theta(2^{t-2p-3})
       +2^{-p-2}\Theta(2^{t-2p-5})\right].
 \end{aligned}
 \label{eq:theta-substitution}
\end{equation}
Use Eq.~\eqref{eqS:JTP} with $q=1/4$.  Its product at $z=1/2$
is $\Theta(1/2)=AQ$: the two positive-factor products run over the odd
powers of $2^{-1}$, and the remaining factor is
$\prod_{j\geq1}(1-2^{-4j})$.
The identities $\Theta(z)=\Theta(z^{-1})$ and
$\Theta(z/16)=4z^{-1}\Theta(z)$, obtained by reversing or shifting the
summation index, reduce Eq.~\eqref{eq:theta-substitution} at
$t=0,\pm1,\pm2$ to
\begin{equation}
 \begin{gathered}
 G_p(0)=1,\qquad G_p(-1)=4G_p(1),\\
 G_p(2)=\frac38,\qquad G_p(-2)=6,
 \qquad p\in\{0,1\}.
 \end{gathered}
 \label{eq:theta-generator-values}
\end{equation}
Since $L_u=4\,2^{u/2}-2^{-u/2}$, its first moment is
$4G_p(1)-G_p(-1)=0$, and its second is
$16G_p(2)-8G_p(0)+G_p(-2)=4$.  Thus, for both parities,
\begin{equation}
 \sum_{u\in\mathbb Z_p}\omega_u=1,\qquad
 \sum_{u\in\mathbb Z_p}\omega_uL_u=0,\qquad
 \sum_{u\in\mathbb Z_p}\omega_uL_u^2=4.
 \label{eqS:theta-moments}
\end{equation}
All three identities are exact and absolutely convergent.

Let $A_\ell:=\prod_{j=1}^\ell(1+2^{-j})$,
$B_\ell:=\prod_{j=1}^\ell(1-2^{-j})$, and
$B_\infty=Q/A$.  For $u\in\mathbb Z_p$, $u\geq-m$, set
$r=(m+u)/2$ and $h=2^{-m/2}$.  Direct cancellation in
Eqs.~\eqref{eq:lambda0}, \eqref{eq:Fr}, and
\eqref{eq:Rc} gives the exact expression
\begin{equation}
 \lambda_r(m)=
 \frac{A}{A_{m-2}}\frac{B_\infty}{B_r}
 \prod_{j=0}^{r-1}(1+2^{j+2-m}).
 \label{eqS:right-lambda-product}
\end{equation}
Uniformly for $|u|\leq m/4$, expansion of the logarithms in
Eq.~\eqref{eqS:right-lambda-product} yields
\begin{equation}
 \lambda_r(m)=1+hL_u+O\!\left(h^2(1+2^{|u|})\right).
 \label{eqS:right-lambda-expansion}
\end{equation}
The implied constant is absolute.  The omitted range can be controlled
together with the eigenvalues, not only with its $\omega$-mass.  Put
$K=(r+1-m)_+$.  Splitting the last product in
Eq.~\eqref{eqS:right-lambda-product} at $j=m-2$, using
$B_\infty/B_r\leq1$, and applying $1+2^j\leq2^{j+1}$ gives
\begin{equation}
 \lambda_r(m)\leq2A^2\,2^{K(K+3)/2}.
 \label{eqS:right-global-bound}
\end{equation}
Combined with the Gaussian exponent in $\omega_u$, this implies, for
$j=0,1,2$ and every fixed $J$,
\begin{equation}
 \sum_{\substack{u\in\mathbb Z_p,\ u\geq-m\\ |u|>m/4}}
 \omega_u\lambda_r(m)^j=o(h^J).
 \label{eqS:right-tail-moments}
\end{equation}
Indeed, for $-m\leq u<-m/4$ the right-hand side of
Eq.~\eqref{eqS:right-global-bound} is bounded.  For $u>m/4$ it is bounded
until $u=m-2$; thereafter its quadratic exponent is at most
$j(u-m+2)^2/8+O(u+m)$, which for $j\leq2$ is strictly dominated by the
$-u^2/2$ exponent of $\omega_u$.  The absent range $u<-m$ similarly has
$o(h^J)$ contribution after multiplication by any fixed power of $L_u$.
In particular, both the omitted amplitude and every moment used below are
$o(h^2)$.

Write $\delta_u=\lambda_r(m)-1$.  Equations
\eqref{eqS:right-lambda-expansion} and \eqref{eqS:theta-moments}, together
with the exact first moment, give
\begin{equation}
 \sum_{u\geq-m}\omega_u\delta_u=o(h^2),\qquad
 \sum_{u\geq-m}\omega_u\delta_u^2=4h^2+o(h^2).
 \label{eqS:right-delta-moments}
\end{equation}
On $|u|\leq m/4$ one has $\sup_u|\delta_u|=o(1)$, and the Taylor remainder
in $\sqrt{1+x}=1+x/2-x^2/8+O(|x|^3)$ is summable there against
$\omega_u$ by the super-Gaussian bound.  Equation~\eqref{eqS:right-tail-moments}
removes the complementary range before this Taylor expansion is used.
Therefore
\begin{align}
 \sum_{r\geq0}a_r(m)
 &=1-\frac18(4h^2)+o(h^2)\notag\\
 &=1-\frac12\,2^{-m}+o(2^{-m}).
 \label{eqS:amplitude-right}
\end{align}
Squaring Eq.~\eqref{eqS:amplitude-right} proves
Eq.~\eqref{eqS:right-asymptotic}.
\end{proof}

\clearpage


\begin{thebibliography}{10}
\providecommand{\url}[1]{#1}
\csname url@samestyle\endcsname
\providecommand{\newblock}{\relax}
\providecommand{\bibinfo}[2]{#2}
\providecommand{\BIBentrySTDinterwordspacing}{\spaceskip=0pt\relax}
\providecommand{\BIBentryALTinterwordstretchfactor}{4}
\providecommand{\BIBentryALTinterwordspacing}{\spaceskip=\fontdimen2\font plus
\BIBentryALTinterwordstretchfactor\fontdimen3\font minus \fontdimen4\font\relax}
\providecommand{\BIBforeignlanguage}[2]{{%
\expandafter\ifx\csname l@#1\endcsname\relax
\typeout{** WARNING: IEEEtran.bst: No hyphenation pattern has been}%
\typeout{** loaded for the language `#1'. Using the pattern for}%
\typeout{** the default language instead.}%
\else
\language=\csname l@#1\endcsname
\fi
#2}}
\providecommand{\BIBdecl}{\relax}
\BIBdecl

\bibitem{AaronsonGottesman2004}
S.~Aaronson and D.~Gottesman, ``Improved simulation of stabilizer circuits,'' \emph{Phys. Rev. A}, vol.~70, p. 052328, 2004.

\bibitem{Gross2006}
D.~Gross, ``{Hudson}'s theorem for finite-dimensional quantum systems,'' \emph{J. Math. Phys.}, vol.~47, p. 122107, 2006.

\bibitem{GrossNezamiWalter2021}
D.~Gross, S.~Nezami, and M.~Walter, ``{Schur--Weyl} duality for the {Clifford} group with applications: Property testing, a robust {Hudson} theorem, and {de Finetti} representations,'' \emph{Commun. Math. Phys.}, vol. 385, pp. 1325--1393, 2021.

\bibitem{HayashiLu2026}
\BIBentryALTinterwordspacing
M.~Hayashi and Y.~Lu, ``Partial stabilizer learning under fixed commuting constraints and the absence of a copy-rate discount,'' 2026, arXiv:2609.13923v1. [Online]. Available: \url{https://arxiv.org/abs/2609.13923v1}
\BIBentrySTDinterwordspacing

\bibitem{Montanaro2017}
\BIBentryALTinterwordspacing
A.~Montanaro, ``Learning stabilizer states by {Bell} sampling,'' 2017, arXiv:1707.04012v2. [Online]. Available: \url{https://arxiv.org/abs/1707.04012}
\BIBentrySTDinterwordspacing

\bibitem{ChangCaroLaroccaWest2026}
\BIBentryALTinterwordspacing
R.~Chang, M.~C. Caro, M.~Larocca, and M.~West, ``Sample-optimal learning of stabilizer states,'' 2026, arXiv:2609.10974v1. [Online]. Available: \url{https://arxiv.org/abs/2609.10974}
\BIBentrySTDinterwordspacing

\bibitem{YuenKennedyLax1975}
H.~P. Yuen, R.~S. Kennedy, and M.~Lax, ``Optimum testing of multiple hypotheses in quantum detection theory,'' \emph{IEEE Trans. Inf. Theory}, vol.~21, no.~2, pp. 125--134, 1975.

\bibitem{EldarMegretskiVerghese2004}
Y.~C. Eldar, A.~Megretski, and G.~C. Verghese, ``Optimal detection of symmetric mixed quantum states,'' \emph{IEEE Trans. Inf. Theory}, vol.~50, no.~6, pp. 1198--1207, 2004.

\bibitem{MontealegreGross2025}
F.~Montealegre-Mora and D.~Gross, ``Duality theory for {Clifford} tensor powers,'' \emph{J. Math. Phys.}, vol.~66, p. 092202, 2025.

\bibitem{FulmanStanton2016}
J.~Fulman and D.~Stanton, ``On the distribution of the number of fixed vectors for the finite classical groups,'' \emph{Ann. Comb.}, vol.~20, no.~4, pp. 755--773, 2016.

\bibitem{ComputationalSupplement}
M.~Hayashi and Y.~Lu, ``Computational supplement to Exact Critical Curve
for Uniform Stabilizer-State Identification,'' version~1.0.0,
Sep.~24, 2026. Ancillary files accompanying this submission,
\path{anc/critical_curve_supplement/}.

\end{thebibliography}
\end{document}